\PassOptionsToPackage{unicode}{hyperref}
\PassOptionsToPackage{hyphens}{url}
\PassOptionsToPackage{dvipsnames,svgnames,x11names}{xcolor}
\documentclass[
  12pt]{article}

\usepackage{amsmath,amssymb}
\usepackage{hyperref}
\usepackage{iftex}
\usepackage{comment}
\usepackage{enumerate}
\usepackage[T1]{fontenc}
\usepackage[utf8]{inputenc}
\usepackage{textcomp} % provide euro and other symbols
\usepackage{lmodern}
\makeatletter
\@ifundefined{KOMAClassName}{% if non-KOMA class
  \IfFileExists{parskip.sty}{%
    \usepackage{parskip}
  }{% else
    \setlength{\parindent}{0pt}
    \setlength{\parskip}{6pt plus 2pt minus 1pt}}
}{% if KOMA class
  \KOMAoptions{parskip=half}}
\makeatother
\usepackage{xcolor}
\makeatletter
\ifx\paragraph\undefined\else
  \let\oldparagraph\paragraph
  \renewcommand{\paragraph}{
    \@ifstar
      \xxxParagraphStar
      \xxxParagraphNoStar
  }
  \newcommand{\xxxParagraphStar}[1]{\oldparagraph*{#1}\mbox{}}
  \newcommand{\xxxParagraphNoStar}[1]{\oldparagraph{#1}\mbox{}}
\fi
\ifx\subparagraph\undefined\else
  \let\oldsubparagraph\subparagraph
  \renewcommand{\subparagraph}{
    \@ifstar
      \xxxSubParagraphStar
      \xxxSubParagraphNoStar
  }
  \newcommand{\xxxSubParagraphStar}[1]{\oldsubparagraph*{#1}\mbox{}}
  \newcommand{\xxxSubParagraphNoStar}[1]{\oldsubparagraph{#1}\mbox{}}
\fi
\makeatother

\usepackage{longtable,booktabs,array}
\usepackage{calc} % for calculating minipage widths
\usepackage{etoolbox}
\makeatletter
\patchcmd\longtable{\par}{\if@noskipsec\mbox{}\fi\par}{}{}
\makeatother
\IfFileExists{footnotehyper.sty}{\usepackage{footnotehyper}}{\usepackage{footnote}}
\makesavenoteenv{longtable}
\usepackage{graphicx}
\makeatletter
\def\maxwidth{\ifdim\Gin@nat@width>\linewidth\linewidth\else\Gin@nat@width\fi}
\def\maxheight{\ifdim\Gin@nat@height>\textheight\textheight\else\Gin@nat@height\fi}
\makeatother
\setkeys{Gin}{width=\maxwidth,height=\maxheight,keepaspectratio}
\makeatletter
\def\fps@figure{htbp}
\makeatother

\makeatletter
\@ifpackageloaded{caption}{}{\usepackage{caption}}
\AtBeginDocument{%
\ifdefined\contentsname
  \renewcommand*\contentsname{Table of contents}
\else
  \newcommand\contentsname{Table of contents}
\fi
\ifdefined\listfigurename
  \renewcommand*\listfigurename{List of Figures}
\else
  \newcommand\listfigurename{List of Figures}
\fi
\ifdefined\listtablename
  \renewcommand*\listtablename{List of Tables}
\else
  \newcommand\listtablename{List of Tables}
\fi
\ifdefined\figurename
  \renewcommand*\figurename{Figure}
\else
  \newcommand\figurename{Figure}
\fi
\ifdefined\tablename
  \renewcommand*\tablename{Table}
\else
  \newcommand\tablename{Table}
\fi
}
\@ifpackageloaded{float}{}{\usepackage{float}}
\floatstyle{ruled}
\@ifundefined{c@chapter}{\newfloat{codelisting}{h}{lop}}{\newfloat{codelisting}{h}{lop}[chapter]}
\floatname{codelisting}{Listing}

\makeatother
\makeatletter
\@ifpackageloaded{caption}{}{\usepackage{caption}}
\@ifpackageloaded{subcaption}{}{\usepackage{subcaption}}
\makeatother

\ifLuaTeX
  \usepackage{selnolig}  % disable illegal ligatures
\fi
\usepackage[]{natbib}
\usepackage{bookmark}

\IfFileExists{xurl.sty}{\usepackage{xurl}}{} % add URL line breaks if available
\hypersetup{
  pdftitle={Title},
  pdfauthor={Author 1; Author 2},
  pdfkeywords={3 to 6 keywords, that do not appear in the title},
  colorlinks=true,
  linkcolor={blue},
  filecolor={Maroon},
  citecolor={Blue},
  urlcolor={Blue},
  pdfcreator={LaTeX via pandoc}}

\definecolor{darkgreen}{RGB}{0,100,0}
\definecolor{darkred}{rgb}{0.64, 0.0, 0.0}

\newcommand{\blue}[1]{{\leavevmode\color{blue}{#1}}}
\usepackage{amsthm}
\usepackage{algorithm}
\usepackage{algpseudocode}
\usepackage{enumerate}
\usepackage{mathrsfs}
\usepackage{mathtools}
\usepackage{bm}
\usepackage{fancybox}
\usepackage{graphicx}
\usepackage{xcolor}
\usepackage{booktabs}
\usepackage{multirow}
\usepackage{subcaption}
\usepackage{fancyhdr}

\usepackage{tikz} % can I use it?
\usetikzlibrary{positioning,arrows.meta} % can I use it?
\DeclareMathOperator*{\argmax}{arg\,max}

\newcommand{\nlim}{\xrightarrow[]{n \to \infty}}
\newcommand{\pnlim}{\xrightarrow[n \to \infty]{P}}
\newcommand{\klim}{\xrightarrow[]{k \to \infty}}

\newcommand{\dnlim}{\xrightarrow[n \to \infty]{d}}

\newcommand{\bmu}{\boldsymbol{\mu}}
\newcommand{\calC}{\mathcal{C}}
\newcommand{\calD}{\mathcal{D}}
\newcommand{\calF}{\mathcal{F}}
\newcommand{\calL}{\mathcal{L}}

\newcommand{\calH}{\mathcal{H}}
\newcommand{\calI}{\mathcal{I}}
\newcommand{\R}{\mathbb{R}}
\newcommand{\N}{\mathbb{N}}

\newcommand{\bigo}{\mathcal{O}}
\newcommand{\smallo}{o}

\newcommand{\bX}{\boldsymbol{X}}

\newcommand{\bY}{\boldsymbol{Y}}
\newcommand{\bR}{\boldsymbol{R}}

\newcommand{\bzet}{\boldsymbol{\zeta}}

\newcommand{\bh}{\boldsymbol{h}}
\newcommand{\bD}{\boldsymbol{D}}

\newcommand{\bS}{\boldsymbol{S}}
\newcommand{\bM}{\boldsymbol{M}}
\newcommand{\ba}{\boldsymbol{a}}
\newcommand{\balp}{\boldsymbol{\alpha}}
\newcommand{\bDel}{\boldsymbol{\Delta}}
\newcommand{\bA}{\boldsymbol{A}}
\newcommand{\bC}{\boldsymbol{C}}
\newcommand{\bZ}{\boldsymbol{Z}}
\newcommand{\bU}{\boldsymbol{U}}

\newcommand{\bP}{\boldsymbol{P}}

\newcommand{\bI}{\boldsymbol{I}}
\newcommand{\beps}{\boldsymbol{\epsilon}}

\newcommand{\bSig}{\boldsymbol{\Sigma}}

\newcommand{\E}{\mathbb{E}}
\newcommand{\colsp}{\mathscr{C}}
\newtheorem{theorem}{Theorem}
\newtheorem{corollary}{Corollary}
\newtheorem{lemma}{Lemma}

\newtheorem{assumption}{Assumption}
\newtheorem{definition}{Definition}

\newcommand{\bH}{\boldsymbol{H}}

\newcommand{\tilX}{\Tilde{\boldsymbol{X}}}

\newcommand{\bgam}{\boldsymbol{\gamma}}
\newcommand{\bzero}{\boldsymbol{0}}
\newcommand{\Prob}{\mathbb{P}}

\newcommand{\anon}{1}
\begin{document}

\def\spacingset#1{\renewcommand{\baselinestretch}%
{#1}\small\normalsize} \spacingset{1}

%%%%%%%%%%%%%%%%%%%%%%%%%%%%%%%%%%%%%%%%%%%%%%%%%%%%%%%%%%%%%%%%%%%%%%%%%%%%%%

\if1\anon
{
  \title{\bf 
  Robust Spatial Confounding Adjustment via Basis Voting} 
  \author{Anik Burman, Elizabeth L. Ogburn$^{\ast}$, Abhirup Datta\thanks{
   Address for correspondence: eogburn@jhu.edu; abhidatta@jhu.edu
    }\hspace{.2cm}\\
    Department of Biostatistics, Johns Hopkins University\\}
    \date{}
  \maketitle
} \fi

\if0\anon
{
  \bigskip
  \bigskip
  \bigskip
  \begin{center}
    {\LARGE\bf Robust Spatial Confounding Adjustment via Basis Voting}
\end{center}
  \medskip
} \fi

\bigskip
\begin{abstract}
Estimating effects of spatially structured exposures is complicated by unmeasured spatial confounders, which undermine identifiability in spatial linear regression models unless structural assumptions are imposed. We develop a general framework for  effect estimation in spatial regression models that relaxes the commonly assumed requirement that exposures contain higher-frequency variation than confounders. We propose \textit{basis voting}, a plurality-rule estimator --- novel in the spatial literature --- that consistently identifies causal effects only under the assumption that, in a spatial basis expansion of the exposure and confounder, there exist several basis functions in the support of the exposure but not the confounder. This assumption generalizes existing assumptions of differential basis support used for identification of the causal effect under spatial confounding, and does not require prior knowledge of which basis functions satisfy this support condition. We design this estimator as the mode of several candidate estimators each computed based on a single working basis function. We also show that the standard projection-based candidate estimator typically used in other plurality-rule based methods is inefficient, and provide a more efficient novel candidate. Extensive simulations and a real-world application demonstrate that our approach reliably recovers unbiased causal estimates whenever exposure and confounder signals are separable on a plurality of basis functions. By not relying on higher-frequency variation, our method remains applicable to settings where exposures are smooth spatial functions, such as distance to pollution sources or major roadways, common in environmental studies.
\end{abstract}

\noindent%

\vfill

\newpage
\spacingset{1.8} % DON'T change the spacing!

\section{Introduction}\label{sec-intro}

% Body of paper.  The number of lines per page (letter-size paper) will be about 26.

Estimating the coefficient corresponding to a structural linear spatial regression of a spatially varying outcome on a spatially varying exposure is a fundamental problem in spatial statistics and environmental epidemiology. However, this task is challenging in the presence of unmeasured spatial variables that simultaneously influence both the exposure (covariate) and the outcome. When such variables are not accounted for, the (conditional) effect size in standard regression models applied to point-referenced spatial data is unidentifiable, leading to biased estimates. This phenomenon is referred to as \textit{spatial confounding}, and it arises frequently in real-world applications where spatial structure is intrinsic to both the exposures and to unmeasured background processes.

An early formal recognition of this issue can be found in \cite{clayton1993spatial}, which demonstrated how including spatially varying error terms in the model changes covariate effect estimates, and this change was favored over the unadjusted estimate which is biased under confounding. Subsequent works have had differing opinions about this \citep{woodard1999estimating,reich2006effects,wakefield2007disease,hodges2010adding,schnell2020mitigating,khan2022restricted,zimmerman2022deconfounding} with arguments put forth both for and against including spatially correlated errors in the model. A particularly influential contribution by \cite{paciorek2010importance} quantified the bias due to spatial confounding in generalized least squares (GLS) settings, relating it to the relative smoothness of the exposure and the unmeasured confounder. This work also underscored the identifiability challenges of separating the exposure effect from the influence of spatial confounding, unless certain restrictions or structural assumptions are imposed. Extensive theoretical and empirical results in \cite{khan2023re} largely agree with the conclusions of \cite{paciorek2010importance} that it is better to include spatial error terms in the analysis model when the exposure is rougher than the unmeasured confounder. \cite{page2017estimation} arrives at similar conclusions for spatial predictions. 
\cite{gilbert2025consistency} proved that the GLS estimator is consistent under varied forms of spatial confounding as long as the exposure has some non-spatial variation. \cite{dupont2022spatial+} proves consistency of their spatial+ method under a similar assumption of Gaussian noise in the exposure. \cite{gilbert2021causal} grounded the spatial confounding problem in the framework of causal inference, going beyond the linear regression case. Their sufficient conditions for identifiability of a causal effect under confounding also assumed some non-spatial variation in the exposure. 

The story is more complicated in settings where the exposure of interest is a smooth function of spatial location, with no non-spatial component.
Estimators like GLS and spatial+, the consistency of which is guaranteed in the presence of noise in the exposure, may no longer work in this smooth-exposure setting. When exposures are \textit{too} smooth relative to the confounder, the exposure effect cannot be identified \citep{bolin2025spatial,datta2025consistent}, but many  methods make progress via the assumption that unmeasured confounders vary mainly at broad spatial scales, while exposures display more localized, high-frequency spatial patterns \citep{thaden2018structural,keller2020selecting,dupont2022spatial+,dupont2023demystifying,guan2023spectral}. In this case the exposure is unconfounded at higher frequencies and, along with the assumption of a linear relationship between the outcome and the exposure, common in spatial regression, this solves the problem of spatial confounding. 

When the unmeasured confounder has the same or finer spatial variation than the exposure, existing techniques generally fail, and standard spatial estimators can perform worse than even the unadjusted estimator as seen in \cite{paciorek2010importance,khan2023re}. This can be viewed through the lens of causal inference as a violation of the \textit{positivity assumption}, which requires that all levels of exposure be observed across levels of confounders \citep{schnell2020mitigating}.  
Examples of exposures that themselves vary at broad spatial scales are common in many environmental studies, e.g., distance to fracking mines \citep{casey2016unconventional,currie2017hydraulic} or to roadways \citep{karner2010near}. \cite{brugge2013highway} explicitly raises concerns about the inability to adjust for co-smooth spatial confounders. Furthermore, the notion of smoothness defined by an ordering of basis functions may not generalize to two or higher-dimensional spatial domains. For example, using the Kronecker product of basis functions of one-dimensional bases creates a two-dimensional basis with two indices dictating smoothness (one along each coordinate);  Approaches like those of \cite{keller2020selecting}, which model the confounder using the lowest indexed basis functions, do not easily generalize for such basis sets. 

In this work, we develop a general framework for effect estimation in spatial linear regression models in the presence of unmeasured spatial confounding, without relying on the assumption that confounding operates more smoothly in space than the exposure. As far as we are aware, ours is the first approach to spatial confounding to relax this assumption. We introduce a  \textit{plurality-rule} assumption inspired by ideas in the robust instrumental variable (IV) literature. The central assumption is that, in a representation of the exposure $X$ and the confounder $U$ along a set of orthonormal basis functions, there are enough basis functions that support the exposure but not the confounder. More precisely, we assume that the set of ratios of the projections of $U$ and $X$ onto each of the basis functions has a larger cluster around zero than around any other value. This assumption subsumes the typical assumption of finer scale variation in the  exposure as a special case. Under the plurality assumption, we show that a set of candidate estimators of the regression effect, one for each basis function, will have the mode at the true coefficient value. We prove that this enables consistent estimation of the effect without prior knowledge of which basis functions support only the exposure. We call the method {\em basis voting}, as each basis yields a candidate estimate, and the most popular candidate -- the mode -- is the final estimate. 

A second contribution of this work is to improve the efficiency of the candidate estimators relative to projection-based estimators that have been used in the IV literature and in a plurality-rule-based deconfounding method recently proposed for time-series data \citep{schur2025decor}. 
We show that in our setting the projection-based candidate estimator is inefficient and present a novel \textit{drop-one} candidate estimator with lower asymptotic variance and finite-sample unbiasedness. We prove consistency and asymptotic normality of the valid  candidate estimators and consistency of the aggregate mode estimator. We empirically demonstrate robustness to the choice of basis functions used for analyzing the data.

The remainder of the paper is organized as follows. Section~\ref{sec:problem-setup} introduces the data-generating mechanism and formalizes the identifiability challenges. Section~\ref{sec:plurality-estimator} presents the proposed basis voting estimator using basis function expansions. Section~\ref{sec:candidates} describes competing methods for constructing candidate estimators and analyzes their properties in large samples.   Section~\ref{sec:simulation} reports results from a comprehensive simulation study comparing our method to existing approaches. Section~\ref{sec:data} showcases an empirical analysis of county-level COVID-19 mortality and air pollution, illustrating the practical implementation and utility of the proposed method. Finally, Section~\ref{sec:discussion} concludes with a discussion of the implications of our findings and directions for future work.

\section{Problem Setup and Assumptions}\label{sec:problem-setup}

Let $\mathcal{D} \subseteq \mathbb{R}^d$ denote a compact spatial domain, and let $\bS_n = \{S_1, \dots, S_n\}$ represent $n$ spatial locations sampled  from a distribution with density $f_S(\cdot)$ on $\mathcal{D}$. At each location $S_i$, we observe an exposure value $X_i \coloneqq X(S_i)$ and an outcome value $Y_i \coloneqq Y(S_i)$. We denote the exposure and outcome vectors by $\bX_n = (X_1, \dots, X_n)^\top$ and $\bY_n = (Y_1, \dots, Y_n)^\top$, respectively. In addition to $X$ and $Y$, we assume there exists an unmeasured confounder $U$ defined at each location, $U_i\coloneqq U(S_i)$, represented by the vector $\bU_n = (U_1, \dots, U_n)^\top$, which is correlated with the exposure and affects the outcome. We assume a univariate exposure, and without loss of generality under the linear model the confounder can be represented by univariate $U$.%Both the exposure and the unmeasured confounder are assumed to be univariate.

In keeping with the existing literature on spatial statistics and spatial confounding, we assume throughout a linear  data generating process (DGP) for the outcome at a generic spatial location $s \in \mathcal{D}$:
\begin{equation}\label{eq:dgp}
Y(s) = \beta X(s) + U(s) + \epsilon(s),
\end{equation}
where $\beta$ is the exposure effect of interest, $\epsilon(s)$ is an independent mean-zero error term, and $U(s)$ is collinear with $X(s)$. If the confounder $U$ were measured, then it would be straightforward to estimate the exposure effect $\beta$, simply by regressing $Y$ in $X$ and $U$ and looking at the estimated slope for $X$. As $U$ is not measured, we cannot directly adjust for it. 
Both $X$ and $U$ are assumed to be deterministic (fixed) continuous functions of the spatial domain from a functional space $\calF(\calD,f_S)$ of continuous functions in $\calD$ equipped with the inner product $\langle g,h\rangle_{f_S}\coloneqq \int_\calD g(s)h(s)f_S(s)ds$ and the corresponding norm $\|\cdot \|_{f_S}$ with respect to the sampling measure $f_S$. Let $\mathcal{H}=\{h_1,h_2, \cdots\}$ denote an infinite-dimensional orthonormal basis of $\calF(\mathcal{D}, f_S)$.  Unlike prior approaches, our method does not require an assumption that the basis functions $h_1,h_2, \cdots$ are ordered by their degree of smoothness or roughness, with $X$ having a longer expansion than $U$. 

Let $X(s) = \sum_{j=1}^\infty \alpha_{x}^{(j)} h_j(s)$ and $U(s) = \sum_{j=1}^\infty \alpha_{u}^{(j)} h_j(s)$. As they are continuous functions on a compact domain, they have finite $L^2$ norms: $\mathbb{E}[X(S)]^2 = \sum_j \left|\alpha_x^{(j)}\right|^2 < \infty$ and $\mathbb{E}[U(S)]^2 = \sum_j \left|\alpha_u^{(j)}\right|^2 < \infty$. By orthogonality, the projections of $X$ on each of the basis functions are given by the coefficients  $\alpha_x^{(j)} = \langle X, h_j\rangle_{f_{S}}$ and similarly for $\alpha_u^{(j)}$. 

Without loss of generality, we partition $\calH$  into three disjoint sets $\calH = \calH_{UX} \sqcup \calH_{X} \sqcup \calH_{U}$ where $\calH_{UX}$ is the set of basis functions shared by both $U$ and $X$, $\calH_X$ is the set unique to $X$ and similarly for $\calH_U$. Figure \ref{fig:dag-spatial} depicts our DGP, with and without the basis expansion. 
\begin{figure}[h]
\centering
% ---  ---  ---  ---  ---  ---  ---  --- 
% Left panel
% ---  ---  ---  ---  ---  ---  ---  --- 
\begin{subfigure}[t]{0.45\textwidth}
\centering
\begin{tikzpicture}[thick, >=Stealth, every node/.style={draw, circle, line width=1pt, minimum size=1cm}]
    % Nodes
    \node (U) at (0,2) {\(U\)};
    \node (X) at (-2,0) {\(X\)};
    \node (Y) at (2,0) {\(Y\)};

    % Arrows
    \draw[->] (U) -- (Y);
    \draw[->] (X) -- (Y);

    % Curved bidirectional arrow between U and X
    \draw[<->, bend left=30] (X) to (U);
\end{tikzpicture}
\caption{Standard confounding setting }
\end{subfigure}
\hfill
% ---  ---  ---  ---  ---  ---  ---  --- 
% Right panel
% ---  ---  ---  ---  ---  ---  ---  --- 
\begin{subfigure}[t]{0.45\textwidth}
\centering
\begin{tikzpicture}[thick, >=Stealth, every node/.style={draw, circle, line width=1pt, minimum size=1cm}]
    % Basis nodes (shaded grey)
  %  \node[fill=gray!30] (H) at (-8,1) {$\calH$};
    \node[fill=gray!30] (HX) at (-6,-1) {$\calH_X$};
    \node[fill=gray!30] (HUX) at (-6,1) {$\calH_{UX}$};
    \node[fill=gray!30] (HU) at (-6,3) {$\calH_U$};

    % Main nodes
    \node (U) at (-4,3) {\(U\)};
    \node (X) at (-4,-1) {\(X\)};
    \node (Y) at (-2,1) {\(Y\)};

    % Arrows from H nodes
   % \draw[->] (H) -- (HX);
    %\draw[->] (H) -- (HU);
    %\draw[->] (H) -- (HUX);
    \draw[->] (HU) -- (U);
    \draw[->] (HUX) -- (U);
    \draw[->] (HUX) -- (X);
    \draw[->] (HX) -- (X);

    % Main causal arrows
    \draw[->] (U) -- (Y);
    \draw[->] (X) -- (Y);
\end{tikzpicture}
\caption{Expanded setting for spatial confounding with basis components (shaded grey)}
\end{subfigure}

\caption{Two graphical representations of the data-generating mechanism: (a) standard confounding and (b) the expanded graph, which shows the basis representation.}
\label{fig:dag-spatial}
\end{figure}
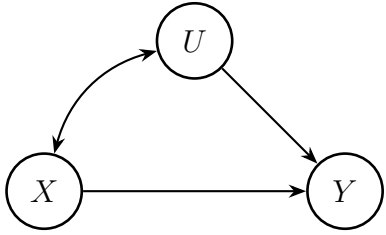
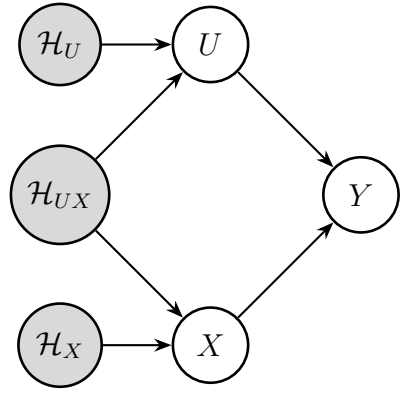
If $\calH_{UX}=\emptyset$, where $\emptyset$ is the null set, then $U$ is orthogonal to $X$ and does not confound the effect of $X$ on $Y$. We therefore focus on the setting where $\calH_{UX}$ is nonempty, representing nontrivial confounding. It may be the case that $\calH_U = \emptyset$, implying that the confounder does not have any independent variation, but we must have $\calH_X \neq \emptyset$ because if not, it is not possible to extract the independent effect of $X$ on $Y$ and thus $\beta$ is unidentifiable. Let us formalize this identification assumption:
\begin{assumption}[Oracle identification assumption]\label{asm:identify}
Let $\balp_x = (\alpha_x^{(1)}, \alpha_x^{(2)}, \dots)$ and $\balp_u = (\alpha_u^{(1)}, \alpha_u^{(2)}, \dots)$ be the sequence of coefficients of $X$ and $U$ for the basis function set $\calH$ respectively, then there exists a non-empty subset $O_{X} \subseteq \mathbb N$ 
such that for all $j \in O_X$, we have $\alpha_u^{(j)} = 0$ and $\alpha_x^{(j)} \neq 0$. 
\end{assumption}

 If the index $j$ of even a single basis function $h_j\in O_X$ were known, that  would suffice to identify $\beta$. This is because the effect of $h_j$ itself on $Y$ is unconfounded, and under linearity, its effect is equal to $\beta$. In practice, however, this will not be the case as $O_X$ is unknown and hence this is only an oracle identification assumption. Previous methods have instead relied on the assumption that the higher frequency basis functions are unconfounded. More formally, let $d$ denote the number of basis functions used in the analysis. Then when using a sufficiently large $d$, these methods essentially assume that the basis functions are ordered in terms of increasing roughness or frequency  $h_1\lessapprox h_2 \lessapprox \cdots \lessapprox h_{d}$, and there exists some $t \leq d \in \N$ such that $\{t+1,\cdots,\} \in O_X$. So, Assumption \ref{asm:identify} is strictly weaker as, rather than requiring unconfounded high-frequency components, we only assume the existence of a non-trivial set $\{h_j: j \in O_X\}$. 
 Figure \ref{fig:coefs-plot} shows two scenarios amenable to our method, one in which the assumption of higher frequency variation in $X$ holds and one in which it is violated. 
\begin{figure}[t!]
    \centering
    \hspace{-1cm} 
    \includegraphics[width=\textwidth]{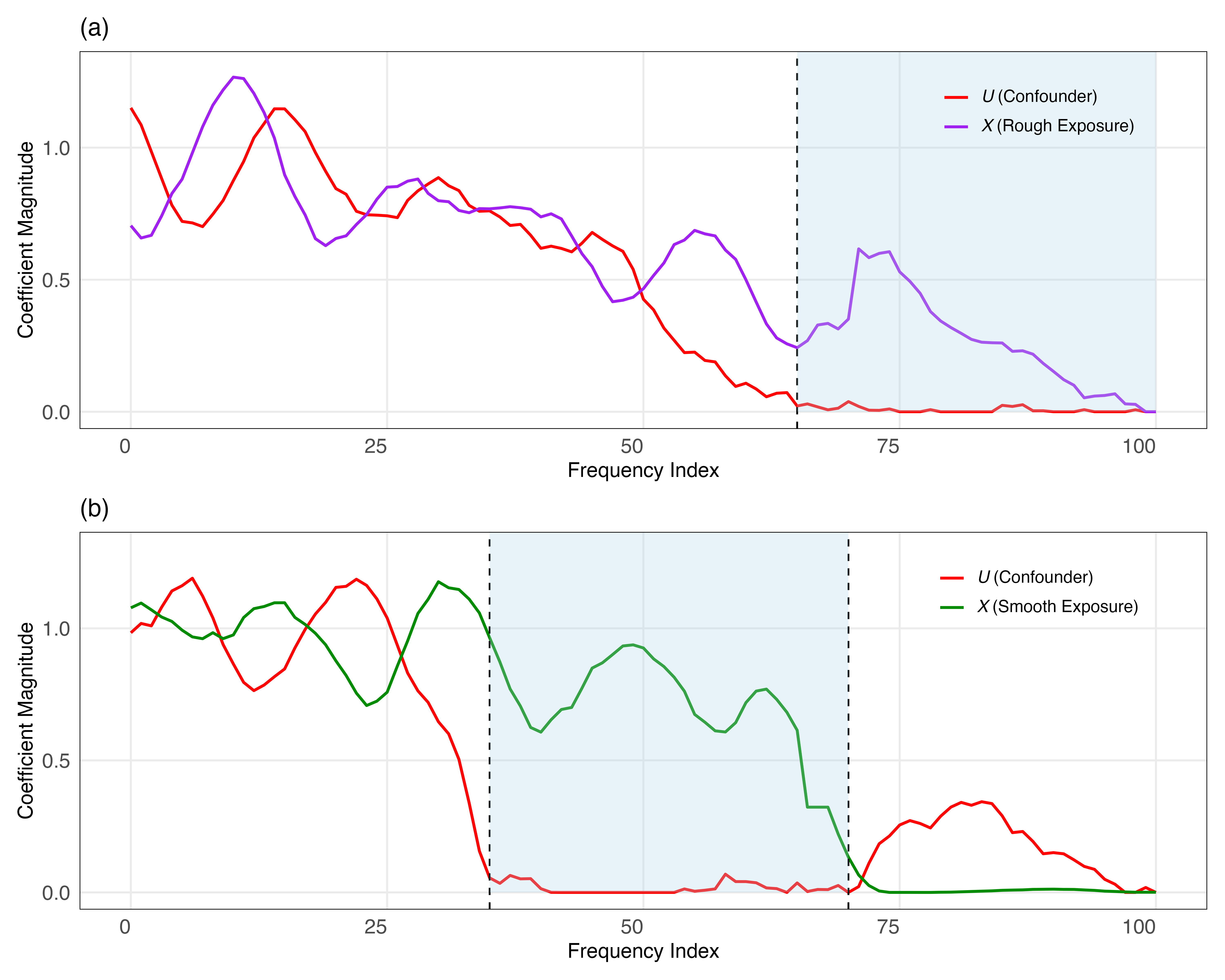}
    \caption{Fourier coefficients of exposures and confounders under varying confounding structures. (a) Exposure $X$ overlaps with unmeasured confounder $U$ at low frequencies, but their behavior diverges at higher frequencies, where $U$ fades and $X$ retains support. (b) $X$ shares low-frequency components with $U$ but decays faster, indicating a smoother exposure than $U$ with an unconfounded region in the middle frequencies. In both Figures, shaded regions mark unconfounded frequency ranges. 
}
    \label{fig:coefs-plot}
\end{figure}
A helpful way to interpret the subset of basis functions $\{ h_j : j \in O_X\}$ is through the lens of \textit{instrumental variable (IV)} methods in causal inference. In IV frameworks, one seeks an instrument $Z$ that satisfies two key properties: (i) it is correlated with the exposure $X$, and (ii) it is independent of the unmeasured confounder $U$ and the outcome $Y$, except through its effect on $X$. When such a valid instrument exists, projecting both $X$ and $Y$ onto the space spanned by $Z$ and regressing the projected outcome on the projected exposure yields a consistent estimate of the causal effect.
In our spatial basis framework, the basis functions $\{h_j\}$ can be thought of as candidate instruments. Under Assumption~\ref{asm:identify}, those $h_j$ with $j \in O_X$ are akin to \textit{valid instruments}, in the sense that they influence $X$ but not $U$, and hence satisfy the classical exclusion restriction. The remaining basis functions $h_j$ for $j \notin O_X$ are \textit{invalid instruments}, since they affect both $X$ and $U$, and their use would introduce bias into causal effect estimation.
 
To illustrate this connection, suppose we knew the set $O_X$, and let $O_{X,d} := O_X \cap \{1,2,\ldots,d\}$ for a sufficiently large $d$ such that $O_{X,d}$ is non-empty. Then a least squares regression of $Y$ onto the basis functions in $O_{X,d}$, which gives a consistent estimator of $\beta$, is equivalent to the standard IV two-stage least squares procedure that projects $X$ onto the span of $\{h_j : j \in O_{X,d}\}$ and regresses $Y$ onto the predicted values.%and $Y$ onto the span of $\{h_j : j \in O_{X,d}\}$, denoting these projections as $P_{X,d}X$ and $P_{X,d}Y$, and fits a simple linear regression to get:
%$$
%\hat{\beta} = %\red{\frac{\langle P_{_{X,d}}X, P_{_{X,d}}Y \rangle_2}{\langle P_{_{X,d}}X, P_{_{X,d}}X \rangle_2}}.
%$$
This equivalence is due to the fact that the projection of $X$ onto the span of $\{h_j : j \in O_{X,d}\}$ simply returns a linear combination of the basis functions in that set.  
This IV-based perspective is consistent with the framework of \citet{woodward2024instrumental}, who show how a variety of existing approaches to unmeasured spatial confounding can be unified under a common IV interpretation. 

The challenge, of course, is that $O_{X}$ is unknown as we no longer assume that the unconfounded variation in $X$ is at the end of the frequency spectrum of the chosen set of $d$ basis functions. Since we do not observe $U$, we cannot directly verify which basis functions fall in the support of $X$ but not $U$.
In the IV literature, several recent approaches have been developed to address this problem under a \textit{majority-rule} \citep[]{kang2016instrumental,windmeijer2019use} or \textit{plurality-rule}  \citep[]{hartwig2017robust,guo2018confidence} assumption, namely that more than 50\% of candidate instruments (majority rule) or that the largest subset of candidate instruments (plurality rule) 
result in estimators with the same unconfounded probability limit.  A similar idea was proposed in the time-series setting with unmeasured temporal confounders by \citet{schur2025decor}. They propose a  \textit{DecoR} estimator that  projects both $X$ and $Y$ onto a collection of basis functions and defines aggregate estimators from these projections based on a majority-rule assumption. 
They frame the problem as an adversarial outlier problem and use the Torrent algorithm \citep{bhatia2015robust} in a robust regression framework to estimate the effect size, based on a user-specified upper bound on cardinality of the set of invalid candidates, which is assumed to be small relative to the sample size. We compare \textit{DecoR}  to our method in a one-dimensional time series simulation setup in the supplementary material (\ref{supsec:decor-vs-bv}) and see that DecoR fails to recover the effect when the majority rule does not hold but the plurality rule does.

Inspired by these ideas, we propose a \textit{plurality-rule estimator} which is defined as the mode of a set of candidate estimators, where each candidate is constructed based on each of the spatial basis functions $\{h_j\}_{j=1}^d$. This approach allows us to recover $\beta$ even without knowledge of $O_{X}$, under a plurality-rule assumption that enough of the candidate basis functions are valid, i.e., contribute to $X$ but not $U$. Most of the existing methods based on majority- or plurality-rule assumptions use a two-stage approach with sample splitting, in which they first select a set of valid candidates and then use them in a second estimation step. In contrast, we propose a mode-based estimator that obviates sample splitting as it is often hard to justify in spatial setting. Instead, our drop-one candidate estimator includes invalid basis functions as precision variables to increase efficiency.  \cite{hartwig2017robust} proposed a similar mode-based estimator with projection-based candidates in the Mendelian randomization IV setting, however they did not provide any theoretical guarantees for their procedure.
In the next section, we formally define the plurality-rule assumption and our estimator, present its theoretical guarantees, and describe its advantages in comparison to existing approaches.

\section{Plurality Rule Mode Estimator}\label{sec:plurality-estimator}

We say that a basis function $h_j$ is \textit{valid} if it gives rise to a candidate estimator $\hat \beta_j$ that is consistent for $\beta$. In order to be valid, the basis function must capture variation in the exposure but not in the unmeasured confounder,  i.e., it must isolate unconfounded exposure variation. This is the case when 
%A basis function $h_j$ is a valid candidate if
$j \in O_X$, i.e.,  $\alpha_u^{(j)}/\alpha_x^{(j)}=0$. On the other hand $h_j$ (and, $\hat \beta_j$) will be invalid if $\alpha_u^{(j)}/\alpha_x^{(j)}=\nu$ for some value $\nu \neq 0$. 

The \textit{plurality rule} assumes that invalid estimators tend to have distinct biases and probability limits, and therefore in the set of probability limits of all candidate estimators, the largest cluster, i.e., the \textit{plurality}, corresponds to the set of valid estimators. In our case, we will show that this assumption clusters the basis functions $h_j$ in terms of the ratio $\nu$, and posits that the biggest cluster is at $0$. The plurality rule does not enforce any minimal cardinality of $O_X$ in terms of the sample size. In fact, $|O_X|$ can be as low as $2$ (provided no two of the remaining ratios are equal).  We define the plurality rule formally below. 

For any $d' \in \mathbb N$, define the set $S_{X,d'} = \left\{1\leq k \leq d': \alpha^{(k)}_x \neq 0\right\}$, that is, the set denoting the indices of the basis functions that have non-zero projection on $X$.  Next, define the set of confounder-exposure projection ratios for the basis functions with indices in the set $S_{X,d'}$ as:
 \begin{equation}\label{eq:ratio}
        R_{d'} = \left\{r_{k}\coloneqq \dfrac{\alpha_u^{(k)}}{\alpha_x^{(k)}}: k \in S_{X,d'} \text{ i.e } \alpha_x^{(k)} \neq 0\right\}
\end{equation}

Let $\calC_{n,d'} = \left\{\hat{\beta}_{k,n}: k \in S_{X,d'}\right\}$ be a class of candidate estimators, where $\hat{\beta}_{k,n}$ corresponds to the basis function $h_k$ such that for each $k\in S_{X,d'}$, $\hat{\beta}_{k,n} \pnlim \beta+r_{k}$. We denote in-probability convergence using $\pnlim$ and weak convergence using $\dnlim$ throughout. Specific candidate estimators will be presented in Section \ref{sec:candidates}, but the main idea of the plurality-rule based estimator is agnostic to this choice. For any $\nu \in \mathbb R$, define
\begin{equation}\label{eq:ratio-value}
    \calI_{\nu,d'} = \left\{r_{k} \in R_{d'}: r_{k} = \nu\right\}
\end{equation}
Then, $k \in O_{X,d'} \iff r_{k} \in \calI_{0,d'}$ and this one-to-one correspondence gives us the rationale for the next assumption, which will help us estimate $\beta$ without the knowledge of $O_X$:
\begin{assumption}[Plurality Rule]\label{asm:estimate}  
    There exists some $d_0\in \N$, such that for all $d^\prime \geq d_0$, $\left|\calI_{0,d^\prime}\right| > \left|\calI_{\nu,d^\prime}\right|$ for any $\nu \neq 0$.
\end{assumption}

The plurality rule  posits that when considering a sufficiently large number of basis functions $h_k$, the cluster of ratios $r_{k}$ at $0$ is larger than clusters of these ratios at any other real number $\nu$. 
This implies that when using  $d \geq d_0$  of basis functions for the analysis, the mode of the set $R_d$ is $0$. To illustrate that this is a reasonable assumption in real world applications, we extracted two meteorological variables--monthly minimum temperature $(^\circ C)$ and total precipitation (mm), for the year 2024-- from the TerraClimate dataset (\cite{abatzoglou2018terraclimate}), which provides monthly climate variables globally at $\sim 4$ km spatial resolution. We averaged them over the 12 months and restricted the domain to the contiguous United States. These variables have an empirical correlation of 0.42, and in an analysis either of them can play the role of an exposure or an unmeasured confounder. Here, we have considered a set of basis functions that are the Cartesian product of two one-dimensional Fourier basis functions and thus lack a natural frequency ordering.  We projected both variables onto this basis function set and looked at both the basis coefficients as well as the mode of $R_d$ for varying choices of $d$. % as it is not possible to know this value for any working basis function class. 
Figure \ref{fig:real-data-ratio-demo} presents the illustration.

\begin{figure}[t!]
  \centering
  \begin{subfigure}[t]{0.9\textwidth}
    \includegraphics[width=\textwidth]
    {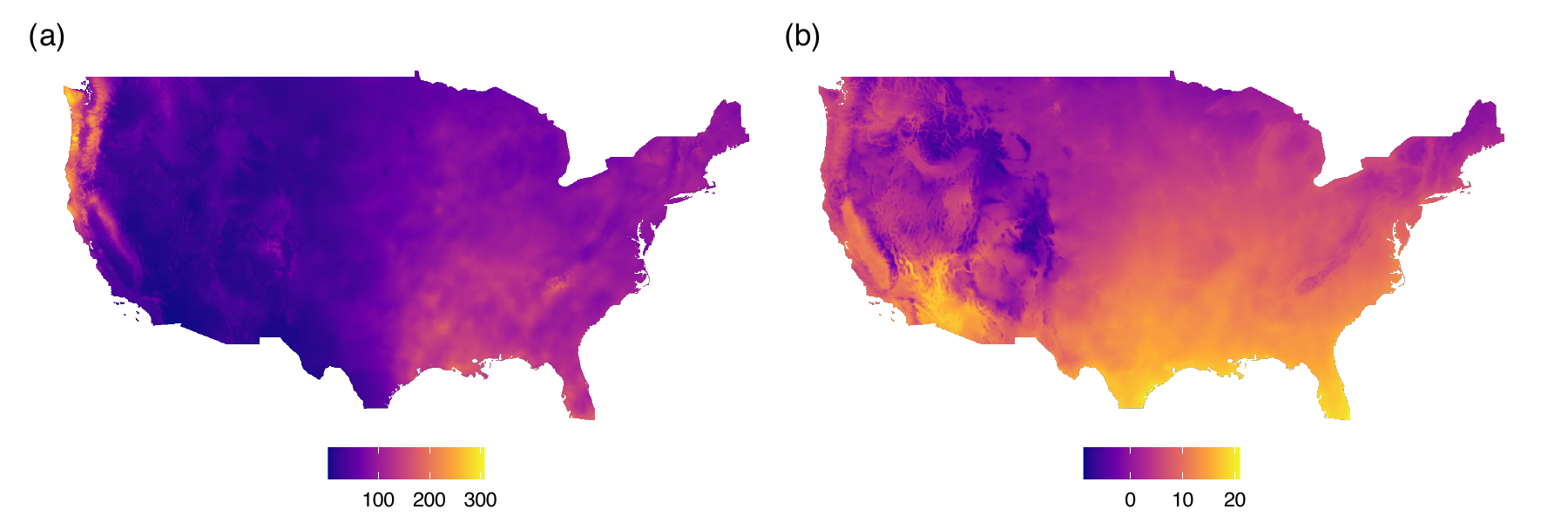} 
  \end{subfigure}
  \hfill
  \begin{subfigure}[t]{0.98\textwidth}
    \includegraphics[width=\textwidth]{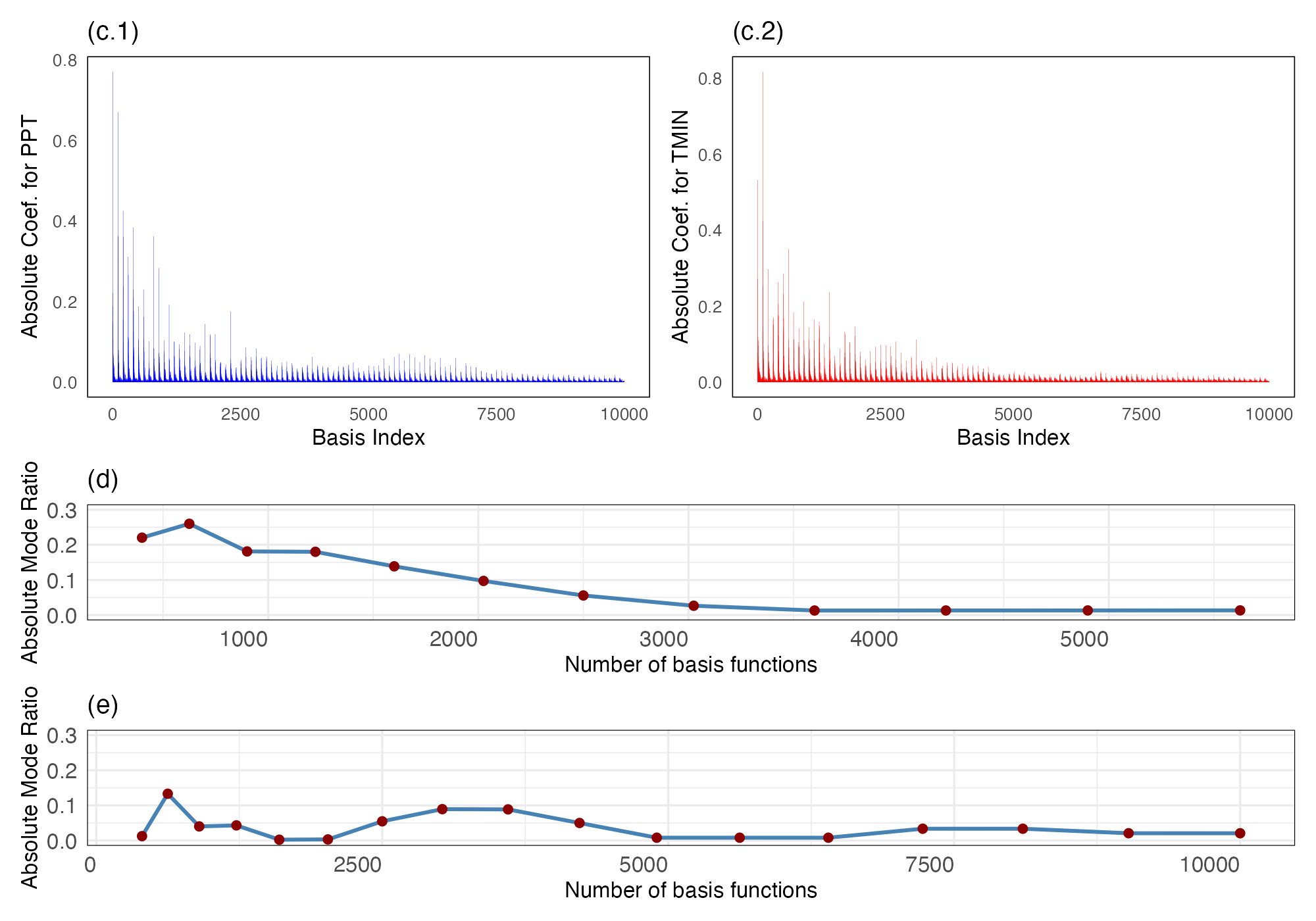}
    % \caption{Caption for the second plot.}
    % \label{fig:subfig2}
  \end{subfigure}
  \caption{(a) Annual mean precipitation (Ppt)  in mm and (b) annual mean minimum temperature (Tmin) in $^\circ C$ for 2024 over the contiguous U.S. Absolute Fourier basis coefficients for (c.1) Ppt and (c.2) Tmin using 10,000 basis functions. Mode of the density of absolute values of the coefficient ratios: (d) Ppt/Tmin and (e) Tmin/Ppt across varying numbers of basis functions.}
    \label{fig:real-data-ratio-demo}
\end{figure}
Comparing Panels (a) and (b), that the map of temperature looks much smoother than that of precipitation. This is  confirmed by the basis coefficient plots in panels (c.1) and (c.2), which show that the support of precipitation has a longer tail. Existing spatial deconfounding methods would thus do well when  precipitation is the exposure and  temperature is the unmeasured confounder, but not the other way around. To assess whether the plurality rule holds in either direction, we plot the modes of $R_d$ as a function of $d$ both ways -- with precipitation as confounder and temperature as exposure in Panel (d), and vice versa in Panel (e). We see that in both plots, the modes decay to $0$ and stabilize as $d$ increases, suggesting that the plurality rule is satisfied in either direction, and estimators based on the plurality rule can account for a wider range of scenarios of spatial confounding.

We formalize the estimation of the mode below. As $\hat \beta_{k,n} \pnlim \beta + r_{d}$, if the mode of $R_d$ is $0$, intuitively the mode of $\calC_{n,d}$ will be $\beta + \text{mode}(R_d) = \beta$. Thus the mode of the candidate estimators can be a good estimator for $\beta$. 
In practice, it is likely that each value of the estimates $\hat \beta_{k,n}$ will be unique, and we use a kernel density smoothing to estimate the mode. 
\begin{definition}\label{def:kde}
    \textit{For a symmetric Kernel function $K$, and a finite set $\ba_m = \{a_1, \cdots, a_m\}$ with each $ a_i\in \mathbb \R$, we can define a Kernel Density, with bandwidth $h$, as:
    $$
    KD(\ba_m,K,h,x) = \dfrac{1}{mh}\sum_{i=1}^m K\left(\dfrac{x - a_i}{h}\right). 
    $$}
\end{definition}
The mode of a KDE based on a finite set of values with unique mode also has the same mode, for small enough bandwidth (See Appendix \ref{lem:plurality-guarantee}). %Hence, we work with the KDE estimate for a small enough bandwidth, which 
Here, the bandwidth will depend on the functional form of the kernel $K$ and the separation from $0$ of any other cluster of the set $R_d$. As the limit points of $\calC_{n,d}$ are $\beta + R_d:= \{\beta + r_k : r_k \in R_d\}$, using shift-invariance of mode\footnote{If $X$ is a random variable with mode at $\mu$, then for any $c \in \mathbb R$, $X+c$ has mode at $\mu+c$.}, we then have  
$\underset{x\in\R}{\argmax} \ KD(\beta + R_{d},K,h,x) = \beta + \underset{x\in \R}{\argmax} \ KD(R_{d},K,h,x) = \beta$. 
This motivates us to propose the following ensemble estimator by combining all the candidate estimators from the class $\calC_{n,d}$: 
\begin{equation}\label{eq:mode_est}
    \hat{\beta}_n = 
    \underset{x\in\R}{\argmax} 
    \ KD(\calC_{n,d},K,h,x).
\end{equation}
We call our method {\em `basis voting'} as the process to select the final estimator $\hat \beta_n$ mimics a voting process. Each basis $h_k$, $k=1,\ldots,d$, represents a voter who votes for a candidate $\hat \beta_{k,n}$, and the most popular candidate (formally, the KDE mode) among the set of candidates $\calC_{n,d} = \{\hat \beta_{k,n}\}$ is the winner, i.e., the final estimator. 
The following theorem gives us theoretical guarantees for the mode estimator. All proofs are in the supplement.
\begin{theorem}\label{thm:mode-consistent}
    Consider the class of estimators $\calC_{n,d} = \left\{\hat{\beta}_{k,n}: k \in S_{X,d}\right\}$ such that for each $k \in S_{X,d}$, $\hat{\beta}_{k,n} \pnlim \beta + r_k$ where $r_k \in R_d$ defined in Equation \ref{eq:ratio}. Then under Assumption \ref{asm:estimate}, for any $d \geq d_0$, a Lipschitz continuous Kernel $K$, and bandwidth $h \leq h_0$ where $h_0$ is some constant depending on $K$ and $\underset{r_k \in R_{d}}{\min} \left|r_k\right|$, 
    with $
    \hat{\beta}_n$ as defined in (\ref{eq:mode_est}), we have  
    $\hat \beta_n \pnlim \beta 
    $.
\end{theorem}
%See the supplementary material for the proof. 
From Theorem \ref{thm:mode-consistent}, we have a consistent estimator of $\beta$ under the plurality rule without any assumption about the relative smoothness of the exposure and confounder. The result holds as long as we are using $d \geq d_0$ number of basis functions. In practice, as $d_0$ is unknown, we recommend repeating the analysis for a range of choice of $d$ and checking that the estimates stabilize after sufficiently large $d$. 

The consistency of the ensemble mode-based basis voting estimator relies on the convergence in probability of the candidate estimators $\hat \beta_{k,n}$ to their respective limits $\beta + r_k$. So, the performance and efficiency of our estimator of $\beta$ requires well-behaved and efficient candidate estimators. The following section presents two such candidate estimators. 

\section{Candidate Estimators}\label{sec:candidates}
The building blocks of the mode estimator $\hat \beta_n$ are the candidate estimators $\hat \beta_{j,n}$ and in this section, we present two sets of candidates. We start with a projection-based estimator similar to two-stage least squares and then we introduce a new drop-one candidate estimator, comparing its performance with the projection-based estimator.  Our goal in this section is to demonstrate how our new drop-one candidate estimator can improve on both bias and efficiency in estimating $\beta$. 

\subsection{Projection-Based Estimator}\label{sec:projn-method}
We begin with a projection-based candidate, similar to two-stage least squares. In the context of time-series deconfounding, this candidate has been used by \citep{schur2025decor} under a special case of the plurality-type assumption. Consider any single basis function $h_j(\cdot)$ with index $j \in S_{X,d}$, indicating that $h_j$ contributes to the spatial variation in the exposure $X$. 
We use this basis function to define a projection-on-projection estimator: For a given set of sampled spatial coordinates $(S_1, \dots, S_n)$, let $\bh_{n,j} = (h_j(S_1), \dots, h_j(S_n))^\top$ denote the vector of evaluated values for the basis $h_j$. The projection matrix onto $\bh_{n,j}$ is given by
$\bP_{n,j} \coloneqq \bh_{n,j}\left({\bh_{n,j}^\top\bh_{n,j}}\right)^{
-1}\bh_{n,j}^\top.$
Then, for any $j \in S_{X,d}$, the projection candidate is computed as the OLS between the projections of the outcome and the exposure, i.e., 
\begin{equation}\label{eq:iv-est}
    \hat{\beta}_{j,n}^{PROJ}\coloneqq \left[\left(\bP_{n,j}\bX_{n}\right)^\top\left(\bP_{n,j}\bX_{n}\right)\right]^{-1}\left(\bP_{n,j}\bX_{n}\right)^\top\bP_{n,j}\bY_{n}.
\end{equation}
Conditional on the sampled locations $\bS_n =(S_1,\cdots,S_n)^\top$,
$\E\left[\hat{\beta}_{j,n}^{PROJ}\right] =
    \beta + \dfrac{\bh_{n,j}^\top\bU_{n}}{\bh_{n,j}^\top\bX_{n}}$. The following theorem provides large sample properties of the estimator. 
\begin{theorem}\label{thm:IV-est-prop}
    Under the DGP as in Equation \ref{eq:dgp} and the estimator as in Equation \ref{eq:iv-est}, assuming a sampling density $f_S > 0$ on the compact domain $\calD$ for the i.i.d. sampled coordinates with $\epsilon(S_i)|S_i$'s being i.i.d with mean 0 and variance $\sigma^2_\epsilon$, defining $r_{j} = \alpha_u^{(j)}/\alpha_x^{(j)}$ for $j$ such that $\alpha_x^{(j)}\neq 0$ 
    % $j \in S_{X,d}$
    , we have:
    \begin{enumerate}[(a)]
        \item $\hat{\beta}_{j,n}^{PROJ}\pnlim \beta + r_j$
        \item $    \sqrt{n}\left(\hat{\beta}_{j,n}^{PROJ}-(\beta+ r_{j})\right) \dnlim N\left(0, \sigma^2_{j}\right)$
            \\where $\sigma^2_{j} \coloneqq \frac{\sigma^2_{U,j} + \sigma^2_\epsilon}{{\alpha_x^{(j)}}^2} + \alpha_u^{(j)}\left(\frac{\alpha_u^{(j)}}{{\alpha_x^{(j)}}^4}\sigma^2_{X,j} - \frac{2}{{\alpha_x^{(j)}}^3}\sigma_{U,X,j}\right)$ for 
$\sigma^2_{U,j}\coloneqq \|h_j U\|^2_{f_S} - {\alpha_u^{(j)}}^2$, $\sigma^2_{X,j}\coloneqq \|h_j X\|^2_{f_S} - {\alpha_x^{(j)}}^2$ and $\sigma_{U,X,j}\coloneqq \langle h_jX, h_jU\rangle_{f_S} - {\alpha_u^{(j)}}{\alpha_x^{(j)}}$.
\end{enumerate}
\end{theorem}
The theorem proves that the projection-based candidate estimators $\hat \beta_{j,n} \pnlim \beta + r_j$ thereby satisfying the requirement to be used for the mode estimator (\ref{eq:mode_est}). It also provides the asymptotic variance of the candidates, which, as we will show, is important in determining efficiency. An immediate corollary from the theorem is the consistency and asymptotic normality of the valid projection candidates. 
\begin{corollary}\label{cor:IV-consistent}
    Under the setup mentioned in Theorem \ref{thm:IV-est-prop}, for any $j \in O_X$ defined in Assumption \ref{asm:identify}, we have $\hat{\beta}_{j,n}^{PROJ} \pnlim \beta$, and after proper scaling, we have the asymptotic distribution  $\sqrt{n}\left(\hat{\beta}_{j,n}^{PROJ}-\beta\right)  \dnlim N\left(0,  \frac{\sigma^2_{U,j}+ \sigma^2_\epsilon}{{\alpha_x^{(j)}}^2}\right)$
    where $\sigma^2_{U,j}$ reduces to $\|h_j U\|^2_{f_S}$.
\end{corollary}
The corollary can be deduced directly from Theorem \ref{thm:IV-est-prop} by setting $\alpha_u^{(j)}=0$ as $j \in O_X$. This implies that if we can choose a valid basis function $h_j$, i.e., $h_j$ is present in the expansion of $X$ and not in the expansion of $U$, then we can get a consistent estimate of $\beta$. Note that, excluding pathological cases, $\sigma^2_{U,j}=\int_\calD (h_j(s)U(s))^2f_S(s)ds - {\alpha_u^{(j)}}^2 = \text{Var}(h_j(S)U(S)) > 0$  even if $\alpha_u^{(j)}=0$. So the asymptotic variance of $\hat{\beta}_{j,n}^{PROJ}$ is strictly greater than $\sigma^2_\epsilon/{\alpha_x^{(j)}}^2$. This inequality has important consequences on the efficiency of the projection estimator, as we will show later.

\subsection{Drop-One Estimator}\label{sec:drop-one-method}
Despite being consistent, the projection-based estimator has finite-sample bias $\frac{\bh_{n,j}^\top\bU_n}{\bh_{n,j}^\top\bX_n}$. Even when $j \in O_{X}$ and the term vanishes asymptotically, it still contributes to the asymptotic variance of $\hat{\beta}_{j,n}^{PROJ}$. We can improve the efficiency of the candidate estimator by including basis functions from the set $\calH_U$ (see right panel of Figure \ref{fig:dag-spatial}). These variables support $U$ but not $X$ and are thus not part of the confounding mechanism. Instead, they are \textit{precision variables}: predictors of the outcome but not exposure, the inclusion of which increases the precision of the regression by explaining some of the variation in the outcome that is not due to or confounded with the exposure. 

We first provide the intuition behind the approach. For simplicity, suppose $U$ admits a finite basis expansion up to $d$ bases. For any $j \in \{1,2,\ldots,d\}$ which supports $X$ but not $U$, i.e., $j \in O_{X,d}$, we \textit{drop} the $j^{th}$ basis function and regress the outcome $Y$ on $X$ and $h_1,h_2,..h_{j-1},h_{j+1},\ldots,h_d$. This adjusts for all confounders, i.e., all $h_k$ which support both $X$ and $U$, and also includes as precision variables all $h_{k'}$ that support only $U$ but not $X$. Leaving out $h_j$ ensures that there is residual variation left in $X$ that is not accounted for the included set of basis functions, adjusting for the confounders ensures that the coefficient corresponding to $X$ is consistent for $\beta$ leveraging this residual variation in $X$ and the linearity of the regression, and including precision variables reduces the variance of this candidate compared to the projection-based estimator.

We now describe the process formally without assuming a true finite basis expansion of $U$. Let $\bH_{1:d} =\begin{bmatrix}
    \bh_{n,1} & \cdots & \bh_{n,d}
\end{bmatrix}$ be the evaluation matrix for the first $d$ basis functions and denote the corresponding projection on its column space by $\bP_d = \bH_{1:d}\left(\bH_{1:d}^\top \bH_{1:d}\right)^{-1}\bH_{1:d}^\top$. Define the projected exposures and outcomes on the first $d$ basis as $\bX_{n,d} = \bP_d\bX_n$ and $\bY_{n,d} = \bP_d\bY_n$. 
For $j \in \{1,\cdots,d\}$, 
define the matrix $
\bH_{-j|1:d} = 
\begin{bmatrix}
    \bh_{n,1} & \cdots & \bh_{n,j-1} & \bh_{n,j+1} & \cdots & \bh_{n,d} 
\end{bmatrix},
$ i.e., the first $d$ basis functions, dropping the $j$-th one. We can decompose $U(s) = \sum_{l=1}^\infty\alpha_u^{(l)}h_l(s) = \alpha_u^{(j)}h_j(s) + \sum_{l \in 1:d, l \neq j} \alpha_u^{(l)}h_l(s) + (1-\lambda_{d}) U_{>d}(s)$ for $1\leq j \leq d$, where $\lambda_{d} = 1$ if $U$ has support only on the first $d$ basis functions and 0 otherwise and hence $U_{>d}(s) = \sum_{l = d+1}^\infty\alpha_u^{(l)}h_l(s)$. Hence 
we can write $\bU_n = \alpha_u^{(j)}\bh_{n,j} + \bH_{-j|1:d}\balp_{-j,u} + (1-\lambda_{d})\bU_{n,>d}$ where $\balp_{-j,u} = \left(\alpha_u^{(1)},\cdots,\alpha_u^{(j-1)},\alpha_u^{(j+1)},\cdots,\alpha_u^{(d)}\right)^\top$. 
Thus, the DGP for the observed values projected onto the first $d$ basis functions can be written as:
\begin{equation}\label{eq:leave-one-dgp}
    \begin{split}
        \bY_{n,d} = \bP_d \bY_n & = \beta \bP_d\bX_n + \bP_d\bH_{-j|1:d}\balp_{-j,u}+ \alpha_u^{(j)}\bP_d\bh_{n,j} + (1-\lambda_{d})\bP_d\bU_{n,>d} + \bP_d\beps_n\\
        & = \beta\bX_{n,d} + \bH_{-j|1:d}\balp_{-j,u}+ \alpha_u^{(j)}\bh_{n,j} + (1-\lambda_{d})\bP_d\bU_{n,>d} + \bP_d\beps_n\\
& = \begin{bmatrix}
    \bX_{n,d} & \bH_{-j|1:d}
\end{bmatrix}\begin{bmatrix}
    \beta\\
    \balp_{-j,u}
\end{bmatrix} + \alpha_u^{(j)}\bh_{j,n} + (1-\lambda_{d})\bP_d\bU_{n,>d} + \bP_d\beps_n\\
& \coloneqq \tilX_{-j}\bgam_{-j,u} + \alpha_u^{(j)}\bh_{n,j} + (1-\lambda_{d})\bP_d\bU_{n,>d} + \bP_d\beps_n\\
    \end{split}
\end{equation}
If $j \in O_{X,d}$, the second term above is $0$. Also, as the columns of $\bP_d$ are asymptotically uncorrelated with $\bU_{n,>d}$, $\bP_d\bU_{n,>d} / n$ vanishes asymptotically. 
Hence, for $j \in O_{X,d}$, we can write the regression as 
$$\E\left[ \bY_{n,d}\right] \approx  \tilde \bX_{-j}\bgam_{-j,u},$$ where the first element of $\bgam_{-j,u}$ is $\beta$. Hence, the first element of the OLS estimate of $\bY_{n,d}$ on $\tilde \bX_{-j}$ is a consistent estimate for $\beta$. Equivalently, if we do a two-stage regression, first regressing $\bY_{n,d}$ on everything but the first column of $\tilde \bX_{-j}$, i.e., on $\bH_{-j|1:d}$, and then regress the residual on the first column of $\tilde \bX_{-j}$, i.e., $\bX_{n,d}$, we should have a consistent estimator of $\beta$. This motivates the following 
drop-one estimator
\begin{equation}\label{eq:drop-est}
    \hat{\beta}_{-j,n} = \dfrac{\bX_{n,d}^\top(\bI_n - \bP_{-j,d})\bY_{n,d}}{\bX_{n,d}^\top(\bI_n - \bP_{-j,d})\bX_{n,d}}, \mbox{ where } \bP_{-j,d} = \bH_{-j|1:d}\left(\bH_{-j|1:d}^\top\bH_{-j|1:d}\right)^{-1}\bH_{-j|1:d}^\top.
\end{equation} 
The following theorem formalizes our claim that by adjusting for $\bH_{-j | 1:d}$, the drop-one candidate estimator both adjusts for confounders and accounts for variation in the outcome due to the precision variables
and is therefore a more efficient consistent estimator than the projection-based estimator.
\begin{theorem}\label{thm:drop-est-prop}
    Consider the DGP as in Equation \ref{eq:dgp} and the estimator as in Equation \ref{eq:drop-est}, assuming a sampling density $f_S > 0$ on the compact domain $\calD$ for the i.i.d. sampled coordinates with $\epsilon(S_i)|S_i$'s being iid with mean 0 and variance $\sigma^2_\epsilon$. Let $d$ be such that $O_{X,d}$ is non-empty. Defining $r_j = \alpha_u^{(j)}/\alpha_x^{(j)}$ such that $j \in S_{X,d}$, we have:
    \begin{enumerate}[(a)]
    \item $\hat{\beta}_{-j,n}\pnlim \beta + r_j$.
    \item For any $j \in O_{X,d}$, we have 
        $\sqrt{n}\left(\hat{\beta}_{-j,n} - \beta\right)   \dnlim N\left(0,\frac{(1-\lambda_{d})\sigma^2_{U,j,>d} + \sigma^2_\epsilon}{{\alpha_x^{(j)}}^2}\right)$
        where $\sigma^2_{U,j,>d} \coloneqq Var(h_jU_{>d}) = \|h_j U_{>d}\|^2_{f_S}$. 
    \item If $\lambda_{d} = 1$, then for $j \in O_{X,d}$ we have $\E\left[\hat{\beta}_{-j,n}\right] = \beta$.
    \end{enumerate}
\end{theorem}
The first two results are similar to those in Theorem \ref{thm:IV-est-prop}, establishing consistency and asymptotic normality of the candidate drop-one estimators. Additionally, if $U$ admits a finite basis expansion with respect to $\calH$, which is implicitly assumed by many existing spatial deconfounding approaches, e.g., the spline-based adjustment of \cite{keller2020selecting}, the third result shows that additionally  $\hat{\beta}_{-j,n}$ is exactly finite-sample unbiased for any $j \in O_{X,d}$. This unbiasedness property is an improvement over the projection-based estimator, which has a finite-sample bias. 

\subsection{Comparison of Drop-one and Projection-based Candidates}\label{sec:drop-one-vs-proj-sd}
We now look at the comparison of asymptotic variances of the estimators based on valid candidate basis functions which support $X$ but not $U$, i.e., $h_j$ for $j \in O_X$. For a finite sample comparison, see the simulations in Section \ref{sec:misspec}. We have seen from Theorems \ref{thm:IV-est-prop} and \ref{thm:drop-est-prop} that for any $j \in O_{X,d}$ respectively, both the candidate's estimates are consistent for $\beta$ and asymptotically normal with mean $\beta$. The following theorem formally states that the asymptotic variance of the drop-one estimator is smaller than that of the projection estimator:
\begin{theorem}\label{thm:asm-var-comparison}
    Consider the DGP as in Equation \ref{eq:dgp} and the estimator as in Equation \ref{eq:drop-est}, assuming a sampling density $f_S > 0$ on the compact domain $\calD$ for the i.i.d. sampled coordinates with $\epsilon(S_i)|S_i$'s being iid with mean 0 and variance $\sigma^2_\epsilon$. Let the basis functions $h_1,h_2,\ldots,$ be such that each $h_j$ has a zero set $\{s \in \calD : h_j(s)=0\}$ of $f_S$-measure zero. Define the candidate estimators $\hat{\beta}^{PROJ}_{j,n}$ and $\hat{\beta}_{-j,n}$ as in Equations \ref{eq:iv-est} and \ref{eq:drop-est} respectively based on basis functions $h_1,\ldots,h_d$. Then for large enough $d$ and $j \in O_{X,d}$, 
    $$
    avar\left(\hat{\beta}_{-j,n}\right)
    <
    avar\left(\hat{\beta}^{PROJ}_{j,n}\right),
    $$
    where $avar(\cdot)$ denotes the asymptotic variance.
\end{theorem}
The  assumption of Theorem \ref{thm:asm-var-comparison}, that the zero set of each of the basis functions are of measure zero, is %not restrictive and are 
satisfied by many of the standard families of basis functions used for analysis in a compact spatial domain with a suitable choice of the sampling density $f_S>0$, such as Fourier basis and polynomial basis classes. The results proves that the drop-one candidate estimator is more efficient than the standard projection-based candidate estimator under a broad range of scenarios. Therefore, we recommend using the drop-one candidates in the basis-voting estimation procedure of Section \ref{sec:plurality-estimator}.

\section{Simulation Study}\label{sec:simulation}
\subsection{Simulation Data Generation}\label{sec:sim-data}
To investigate the performance of the proposed estimator under controlled conditions, we have conducted a comprehensive simulation study. The primary goal of the simulation was to evaluate whether basis voting could recover the causal effect $\beta$ in the presence of unmeasured spatial confounding, when the structure of confounding aligns with our plurality assumption.

We considered a spatial domain $\mathcal{D} = [0,1]^2$, where the spatial coordinates were chosen as a grid inside it. For the main simulations, we used a regular grid design with $n = 30^2 = 900$ spatial locations. At each location, both the exposure $X(s)$ and the unmeasured confounder $U(s)$ were generated as linear combinations of a common set of spatial basis functions, constructed from sine and cosine transforms of spatial coordinates, i.e., the Fourier basis functions. This construction allowed us to control the smoothness and spatial scale of both the exposure and the confounder. The coefficients for $X$ and $U$ on the basis functions were chosen to emulate realistic confounding while ensuring that the ratio $\alpha_u^{(j)} / \alpha_x^{(j)}$ has a well-separated mode at zero. Specifically, the coefficients were designed such that a subset of basis functions contributed to both $X$ and $U$ (i.e., confounded directions) while some basis functions were exclusive to either $X$ or $U$.

The non-zero values of the exposure coefficients vector $\left\{\alpha_x^{(j)}\right\}$ and confounder coefficients vector $\left\{\alpha_u^{(j)}\right\}$ were sampled independently from $Unif(3,6)$ distribution along with random sign $\{-1,1\}$ assigned with equal probability. The coefficients were kept fixed over all the data replicates. 
The magnitude of the confounding bias in naive OLS was calibrated based on the inner product of the coefficient vectors, ensuring a challenging identification task. 

The observed outcome was generated using the structural equation 
$Y(s) = \beta X(s) + U(s) + \epsilon(s)$
where $\beta = 2.5$ is the true causal effect, and $\varepsilon(s) \sim \mathcal{N}(0, \sigma^2_\varepsilon)$ with $\sigma_\epsilon = 0.1$ is i.i.d. noise. Based on whether $X$ is smoother than $U$ or not, there are 2 different DGP scenarios that we have considered. 

\subsection{Correct Basis Specification}
We simulate two distinct scenarios: 
(a) the exposure is a rougher function of space than the confounder, and (b) the exposure is smoother. Many existing methods assume the exposure is spatially rougher than the confounder, and we aim to assess the robustness of all methods when this assumption is violated.

\begin{figure}[ht]
\centering
\includegraphics[width=0.95\textwidth]{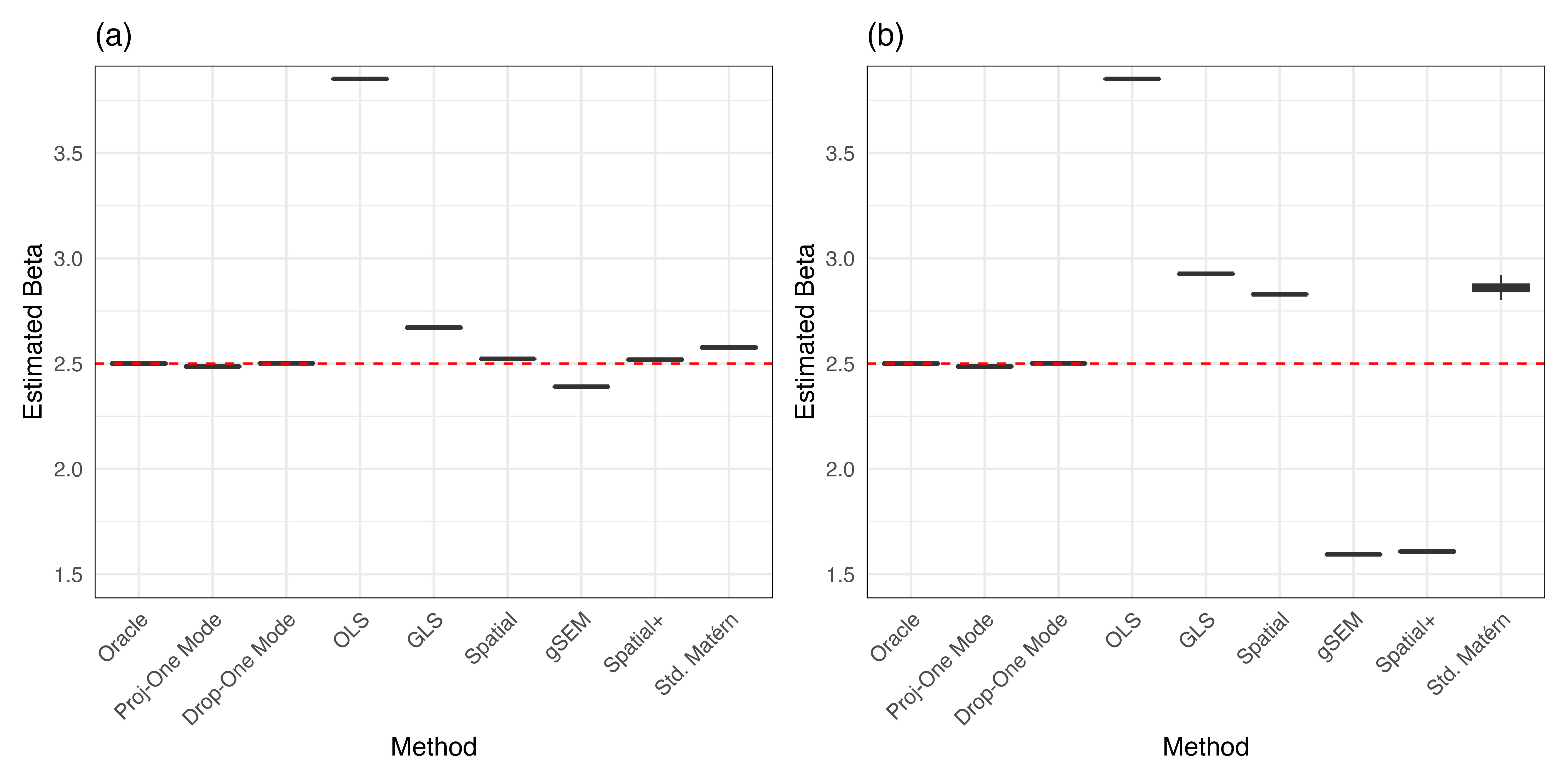}
\caption{Performance of different methods for estimating $\beta = 2.5$ under two scenarios: (a) when the exposure $X$ is rougher than the unmeasured confounder $U$ and (b) when $X$ is smoother than $U$,
using the correct basis specification. Our mode-based estimators (named {\em Proj-One Mode} and {\em Drop-One Mode}) recover the effect sizes in both cases; competitors bias when $X$ is smoother.
}
\label{fig:est-correct}
\end{figure}

Figure~\ref{fig:est-correct} compares our proposed estimator to several existing methods, using the oracle estimator as a benchmark (which leverages knowledge of the true $U$ and the shared basis functions for optimal efficiency). Competing methods include Ordinary Least Squares (OLS), Generalized Least Squares (GLS) with a pre-specified exponential covariance kernel using the BRISC R-package \citep{saha2018brisc}, the \textit{Spatial} basis adjustment approach using TPRS basis set in a similar essence to \citet{keller2020selecting} by adding a spline term in the regression (to adjust for the unmeasured confounder $U$) additional to the exposure of interest, the gSEM method of \citet{thaden2018structural}, the Spatial+ approach of \citet{dupont2022spatial+}, and the Standard Matérn model-based spectral domain adjustment of \citet{guan2023spectral}. The results reveal important differences in method performance. When the exposure is rougher than the unmeasured confounder with respect to the Fourier basis (Panel (a)), nearly all methods, including ours, accurately recover the true value of $\beta = 2.5$. The only exception is the unadjusted OLS estimator, which of course is expected to be biased under confounding. When the exposure is smoother than the confounder (Panel (b)), all existing methods yield substantially biased estimates, while our mode-based estimators --- both project-on-one and drop-one approaches -- are almost unbiased and accurately estimate $\beta$. This experiment highlights the robustness of our method to violation of the rough exposure assumption. 

\subsection{Robustness to Basis Misspecification}\label{sec:misspec}

The experiments in the previous section were in the idealized scenario in which the basis functions used for estimation match those that generated the data. Using the correct basis is not a requirement for our method, which only requires that the plurality holds under whichever basis set is used in the analysis to represent the exposure and the confounder. Here, we assess the robustness of our method to basis misspecification.

To explore the effect of basis choice, we consider two misspecified alternatives for analyzing the data: (a) standard thin plate regression splines (TPRS), orthogonalised over the set of locations, and (b) eigen functions derived from the empirical covariance structure of the coordinates, where the latter are obtained by eigen decomposition of a Mat\'ern covariance kernel with fixed parameters.

Figure~\ref{fig:coef-ratio-miss} displays the ratios of the projected coefficients for the confounder and exposure, $\alpha_{u}^{(j)}/\alpha_{x}^{(j)}$, onto the $j$-th misspecified TPRS basis function, under two smoothness regimes for $X$ and $U$.
\begin{figure}[ht]
    \centering    \includegraphics[width=0.95\textwidth]{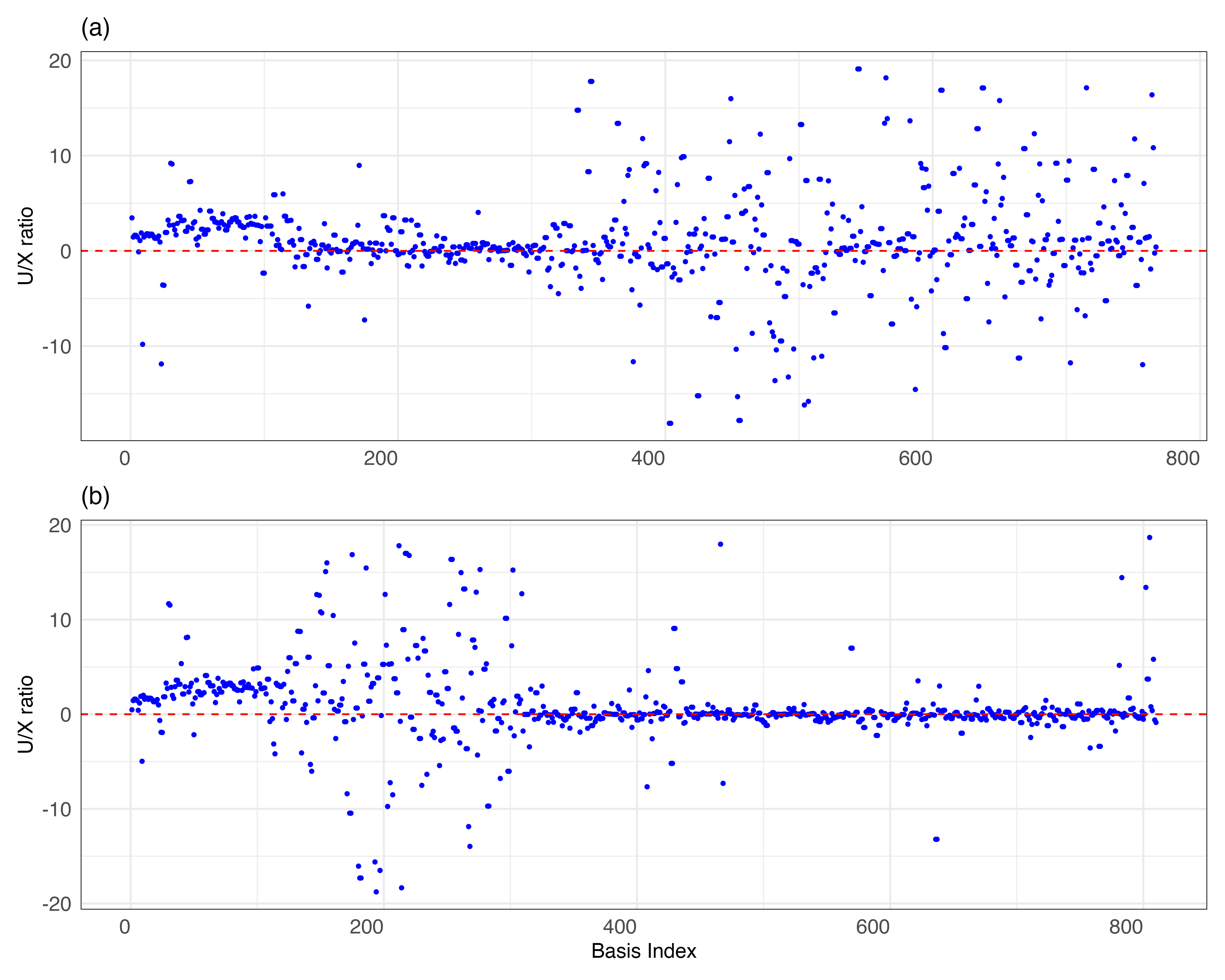}
    \caption{Estimated ratio $\alpha_{u}^{(j)}/\alpha_{x}^{(j)}$ for a misspecified TPRS basis under two smoothness scenarios: (a) $X$ is smoother than $U$, and (b) $X$ is rougher than $U$.}
    \label{fig:coef-ratio-miss}
\end{figure}
As we see, in both scenarios, the coefficient ratios have the biggest clustering around the $y=0$ line, indicating that the modal value remains close to zero. This supports the validity of Assumption~\ref{asm:estimate} for this class of basis functions which was not used to generate the data, regardless of whether the exposure is smoother or rougher. Similar findings were obtained for the eigen basis functions. 

Figure~\ref{fig:basis-misspecified} presents the estimated effect sizes produced by our method when using the first $d$ 
basis functions (with $d$ varying from 150 to 850, for a sample size of 900 observations), as the true value of $d_0$ from Assumption \ref{asm:estimate} is unknown, especially under basis misspecification. Results are presented for both TPRS and eigen basis sets, and under both relative smoothness scenarios for $X$ and $U$. By varying $d$ from small values up to nearly $n$, we observe that the estimator stabilizes as the number of candidate basis estimates increases and the plurality rule becomes more clearly satisfied. Evaluating the estimator across a range of $d$ therefore provides a practical safeguard against the instability that can arise when $d$ is chosen to be too small for the plurality condition to hold.

\begin{figure}[http]
    \centering    \includegraphics[width=0.95\textwidth]{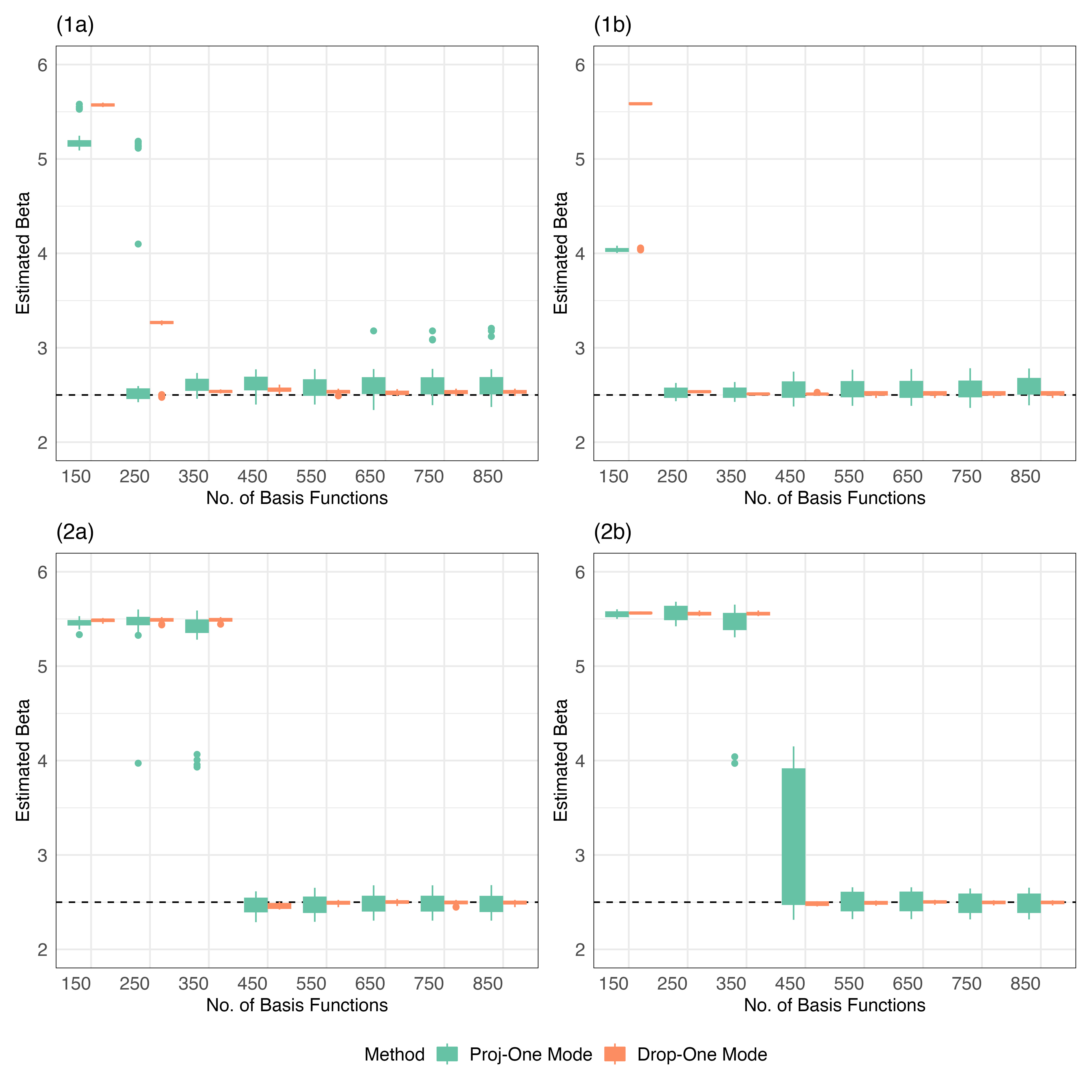}
    \caption{Performance of the two proposed methods under two smoothness scenarios: (1) $X$ is smoother than $U$, and (2) $X$ is rougher than $U$, each evaluated with two misspecified basis function types: (a) TPRS basis and (b) Eigen basis functions. The black dotted line is the true effect size $\beta = 2.5$.}
    \label{fig:basis-misspecified}
\end{figure}

These results show that even under basis misspecification, stable estimation of the effect size is attained once a sufficient number of basis functions are included in the model for the plurality rule to kick in. In terms of bias, the projection method performs worse as compared to the drop-one methods in finite samples. Also, the projection method exhibits a greater standard error. Both of these trends are consistent with the theoretical findings in Section~\ref{sec:drop-one-vs-proj-sd}. We can see that as the number of basis functions in the model increases, the standard error of the estimates also increases, illustrating a bias-variance tradeoff in the number of basis functions included: more basis functions will generally reduce the bias of the estimate, but at the cost of increasing the standard error due to more parameters being estimated. Overall, the drop-one candidates based mode estimator is much more robust to the number of basis functions used. The experiments in this section show that our mode estimator works well even with a misspecified basis, both when the exposure is rougher or smoother, and that using the drop-one candidate estimators generally leads to better performance than the projection-based candidate estimators.  
\section{Real Data Analysis}\label{sec:data}

We illustrate the robustness of our mode estimator in an analysis studying the effect of air quality on COVID-19 mortality. We utilize the county-level COVID-19 dataset analyzed in \cite{guan2023spectral}  and originally collected by \cite{wu2020air}, which includes comprehensive information for $n=3109$ U.S. counties. For each county, the outcome is the cumulative number of COVID-19 deaths through May 12, 2020 (originally obtained from Johns Hopkins University Coronavirus Resource Center), scaled by the county’s population. The primary exposure variable is the long-term average PM\textsubscript{2.5} concentration, computed by averaging estimates from an established exposure prediction model over the years 2000–2016. Figure~\ref{fig:real-data-exp-outcome} displays the scaled death counts and PM\textsubscript{2.5} concentrations for each county. Counties with zero death counts are shown in gray in the figure. 
\begin{figure}[http]
    \centering
    \includegraphics[width=0.95\textwidth]
    {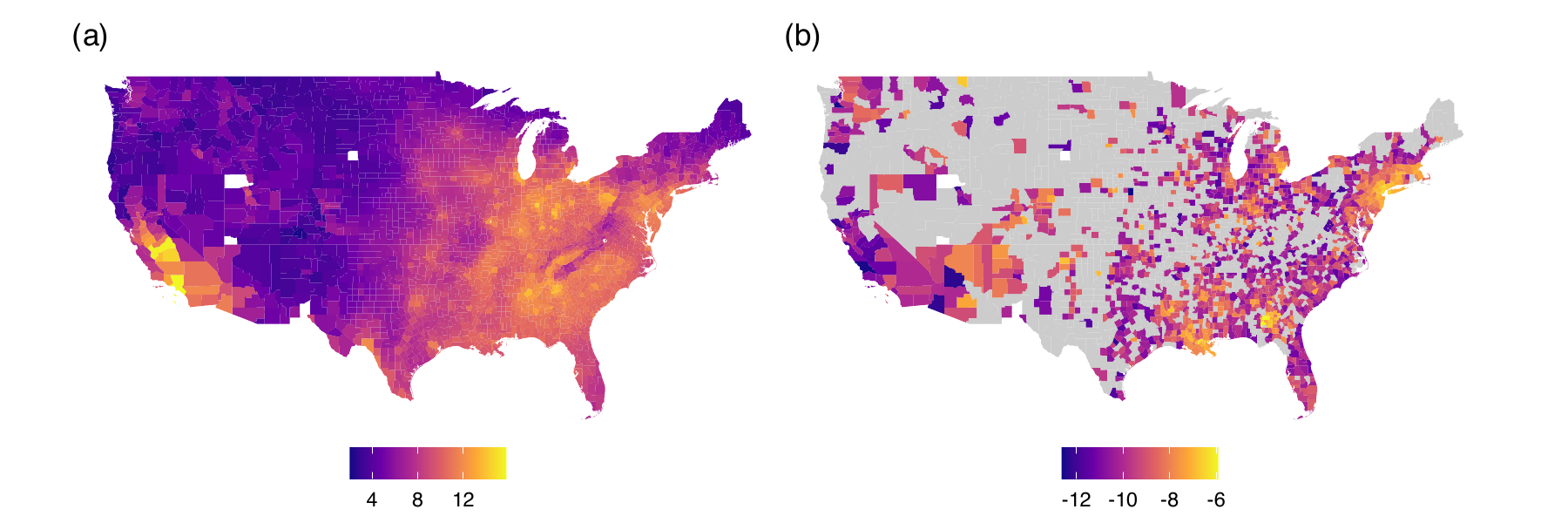}
    \caption{Maps showing (a) Average PM\textsubscript{2.5} in $(\mu g/m^3)$ for the time period 2000-2016, and (b) the cumulative log COVID-19 mortality rate through May 12, 2020, scaled by population. }
    \label{fig:real-data-exp-outcome}
\end{figure}

In addition to the main exposure and outcome variables, the data set contains 20 additional county-level covariates which are potential confounders. To transform this data into a standard geostatistical one, we assign each county the coordinates of its geographic centroid. Counties for which the boundary centroids could not be calculated were dropped from the analysis resulting in an effective sample size of $n=3082$. Consistent with our focus on linear models, we apply a logarithmic transformation to the scaled mortality outcome, replacing zeros with a small positive value to avoid loss of observations and ensure stability of the log transformation. The ordinary least square estimate of the mortality rate ratio for 1 unit increase in PM\textsubscript{2.5} even after adjusting for the measured confounders is $\exp(-1.286) = 0.267$ which shows a counterintuitive result that higher pollutant levels are protective against COVID-19 mortality. This paradoxical finding is also discussed in \cite{wu2020air} and \cite{guan2023spectral}, and we agree with their conclusion that there is a possibility for the presence of unmeasured spatial confounding.

According to  \cite{guan2023spectral}, after adjusting for spatial confounding, the estimated percentage increase in mortality rate associated with a one-unit increase in PM\textsubscript{2.5} is approximately 16\%, i.e., the effect size is around $1.16$. To assess the robustness of our proposed approach relative to existing methods that assume an ordered set of basis functions for adjustment, we consider two classes of spatial basis functions: (1) the conventional thin-plate regression splines (TPRS) orthogonalised over the set of locations which exhibit an inherent ordering of smoothness, and (2) the cartesian product of two one-dimensional Fourier bases, which lack such canonical ordering. As a benchmark, we compare our mode-based estimator to the standard basis adjustment approach in regression, which is known to yield consistent results when basis functions are ordered by increasing roughness. 

\begin{figure}[ht]
    \centering
    \includegraphics[width=0.95\textwidth]{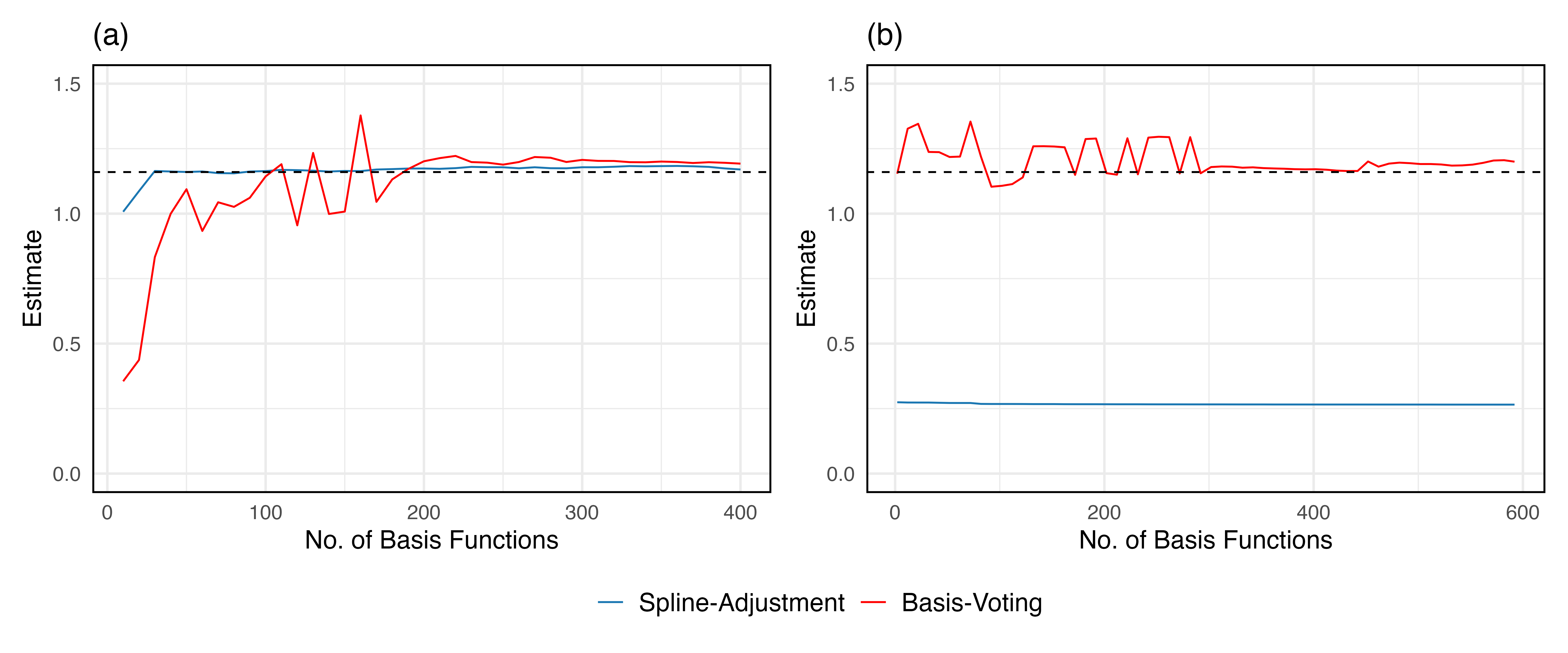}
    \caption{Results for the COVID-19 Data: Estimated mortality-rate ratio for an increase of $1 \mu g/m^3$ increase in PM\textsubscript{2.5} for different numbers of basis functions. The results are for two sets of basis functions: (a) when the basis functions have a canonical ordering of increasing roughness, and (b) when the basis functions do not have a canonical ordering among the basis functions. The blue line is for the model adjusted for $k$ basis functions along with the measured covariates, and the red is for basis voting, our mode-based estimator. The black dotted line is the estimate of 1.16 or 16\% reported by \cite{guan2023spectral}.}
    \label{fig:real-data-figure}
\end{figure}

For our analysis, we first regress out the additional covariates from both the exposure and the outcome to enable direct application of our method. Figure~\ref{fig:real-data-figure} compares our estimator to the standard basis adjustment, with the left panel corresponding to the TPRS basis and the right panel to the cartesian product of Fourier basis. 

As anticipated, when the basis set has an ordering of smoothness (as with TPRS), the results of both estimators are highly similar and align with the estimate from the previous analyses. Notably, when switching to the unordered Fourier basis, our mode-based estimator continues to produce estimates close to the presumed effect, whereas the standard basis adjustment method fails to control for unmeasured spatial confounding, producing a very different estimate. This result underscores the advantage of our approach in settings where the exposure does not represent a `rougher' function of space than the confounder in terms of a given basis expansion. 

\section{Conclusion and Discussion}\label{sec:discussion}

Unmeasured spatial confounding poses a major challenge for effect size estimation in spatial linear models, particularly when exposures and confounders share smooth spatial structure. Existing approaches typically assume exposures have some non-spatial noise, or at the very least, vary more \textit{roughly} than confounders, but this assumption breaks down in many environmental health settings where exposures themselves are smooth or when using basis functions where there is no single notion of smoothness.

We address this gap by modeling exposures and confounders through basis expansions and show that the effect is identified under a plurality rule assumption on the basis functions in the support of the exposure but not the confounder. Leveraging this structure, we introduce basis voting, a plurality-rule based process to arrive at an estimator that identifies causal effects without requiring prior knowledge of valid candidate bases or bounds on the number of invalid ones. We provide theoretical guarantees for consistency and propose a provably more efficient candidate estimator with improved asymptotic variance. Simulations and a real data analysis confirm that our method delivers unbiased estimates whenever exposure and confounder signals are separable on a plurality of basis functions, regardless of their relative smoothnesses. 

Plurality-rule based estimation provides a new, flexible and robust tool for addressing spatial confounding in complex applied settings, and this work advances spatial causal inference by relaxing the relative smoothness assumptions on which previous work relies.  While our framework is limited to linear effects, we plan to explore future extensions to nonlinear models. Also, developing valid statistical inferential procedures using the mode-estimator is an important future direction.  

% \pink{AD: I have not read the supplement or proofs. Please make sure they are correct.}

\addtolength{\textheight}{-.2in}

\section{Acknowledgement} The work was partly supported by grants from the National Institutes of Health and the Office of Naval Research. %National Institute of Health Sciences grant R01 ES033739 and ONR grant N000142412701. 
The authors acknowledge the use of ChatGPT.com for improving the exposition in some parts, code writing, and for some assistance with developing the theoretical results (some claims generated by the authors were proved with help from ChatGPT). All code and proofs were independently verified by the authors, who take full responsibility for the content. We also sincerely thank Professor Jonas Peters and Felix Schur for detailed and valuable feedback on an initial draft of the manuscript. 

\bibliography{bibliography2.bib}

\newpage
\renewcommand{\thesection}{S\arabic{section}} 
\setcounter{section}{0}
\renewcommand{\thelemma}{S\arabic{section}.\arabic{lemma}}
\setcounter{lemma}{0}

\makeatletter
\renewcommand{\theHsection}{S\arabic{section}}
\renewcommand{\theHlemma}{S\arabic{section}.\arabic{lemma}}
\makeatother

\phantomsection\label{supplementary-material}
\bigskip

\begin{center}

{\large\bf SUPPLEMENTARY MATERIAL}

\end{center}

\section{Comparison of Basis Voting and DecoR}\label{supsec:decor-vs-bv}

We present a comparative evaluation of our proposed \emph{basis voting estimator} against the \emph{DecoR} method introduced by \cite{schur2025decor} within a time-series framework. Although designed in a one-dimensional framework, DecoR is conceptually the closest to our basis voting framework. However, some aspects of the underlying philosophies behind the two methods differ substantially. In addition to contrasting their finite-sample behavior, we illustrate how basis voting naturally accommodates settings in which the number basis $d$ used in the analysis is not fixed, but instead allowed to increase with the sample size $n$. This flexibility is important as in practice, %al applications, 
the number of basis functions used in the analysis typically grows with sample size. %number of basis functions. % effective number of basis components that grows with $n$.

DecoR formulates the confounding problem as an adversarial outlier problem after transforming the exposure and outcome into the frequency domain by projecting them onto all $n$ basis functions. Its performance crucially depends on the parameter $a$, which encodes prior knowledge about the proportion of transformed data points that must be treated as outliers---equivalently, the proportion of basis functions that generate invalid candidate estimates. The accuracy of DecoR therefore hinges on approximately correctly specifying $a$, making this choice critical in practice. Consistent with the modeling assumptions in \cite{schur2025decor}, we use the correct basis specification for both DecoR and basis voting to ensure a fair comparison.

Figure \ref{fig:decor-vs-mode} displays simulation results in a regime where the proportion of invalid basis-specific estimates increases with $n$. For our method, the mode is computed using all $n$ basis-specific regression estimates, directly highlighting that basis voting remains applicable even when $d = d(n)$ grows with sample size rather than remaining fixed. %We exclude the noiseless case, as our data-generation mechanism throughout the manuscript assumes additive non-zero noise. 
Across all simulations, the true parameter value is $\beta = 2.5$. We consider three representative confounding scenarios:

\begin{enumerate}
    \item[(a)] \textbf{Majority rule violation with homogeneous outliers.} \\
    In this scenario, the majority of basis-specific estimates are contaminated, and all outliers yield the same biased value. 
    As these outliers are majority, the mode is not at $\beta$ but at this biased value. So this is a scenario where both the majority and plurality conditions fail, and we therefore expect both DecoR and basis voting to be inconsistent. This case serves to %as a negative control 
    illustrate%ing that 
    when either methods cannot recover $\beta$.

    \item[(b)] \textbf{Majority rule holds with a fixed linear proportion of homogeneous outliers.} \\
    A constant fraction of the basis functions produces the same incorrect estimate, but the uncontaminated basis functions (or valid candidates, in our terminology) still form a strict majority. Here, DecoR’s error increases with the error variance (as mentioned in \cite{schur2025decor}), reflecting its sensitivity to noise and its reliance on the tuning parameter $a$. In contrast, plurality rule holds for this setting as the number of valid basis functions are the majority. We see this in the results where basis voting remains stable and demonstrates consistency across sample sizes.

    \item[(c)] \textbf{Plurality (but not majority) rule holds with linearly growing outliers.} \\
    For each $n$, a fixed proportion $p$ of the basis functions are outliers, and although they do not form a majority, their number grows linearly with $n$. The plurality rule holds while the majority rule fails. In this setting, basis voting is consistent whereas DecoR is not. This scenario demonstrates that the truncation threshold $d$ in Assumption \ref{asm:identify} can grow with $n$, provided that the proportion of valid basis-specific estimates remains the largest group.
\end{enumerate}

\begin{figure}[ht]
\centering
    \includegraphics[width=0.8\textwidth]{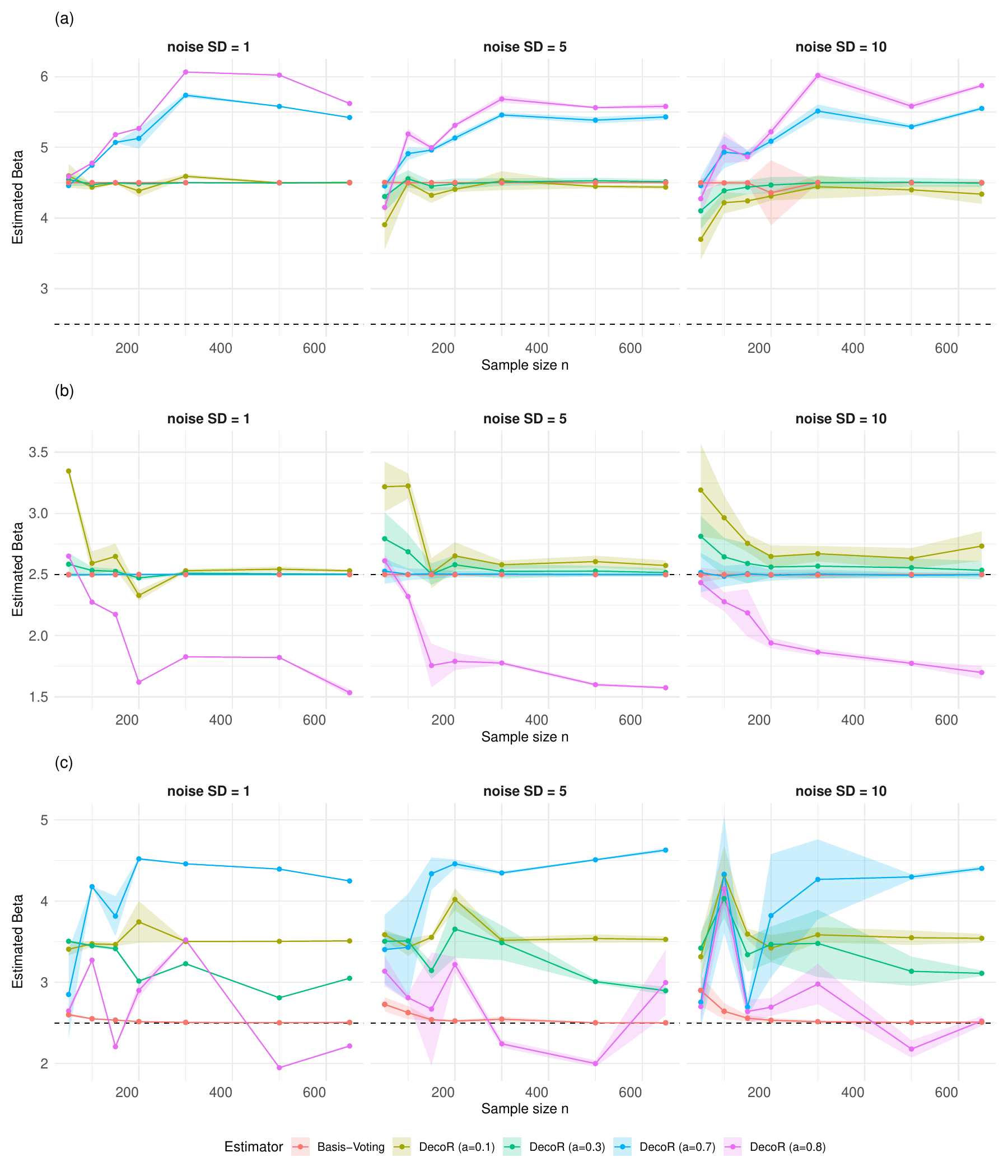}
    \caption{Basis-Voting vs DecoR performance across the three spatial confounding scenarios: (a) 60\% bias of +2, 40\% bias of 0, (b) 70\% bias of 0, 30\% bias of +2, (c) 35\% bias of 0, 25\% bias of +1, 20\% bias of +2, 20\% bias of +3, each evaluated under error variances $\sigma \in \{1,5,10\}$.}
    \label{fig:decor-vs-mode}
\end{figure}

Together, these experiments highlight the distinct philosophies underlying the two approaches. DecoR requires specifying the proportion of outliers and casts confounding as an adversarial contamination problem, whereas basis voting does not rely on such prior knowledge and remains robust under both basis misspecification and growing number of basis functions used in the analysis as long as the plurality rule holds. These results underscore the flexibility and generality of the basis voting framework.
\section{Notations}
We denote $\R$ and $\N$ to be the set of real numbers and natural numbers, respectively. For a sequence of random variables, $\pnlim$ denotes convergence in probability and $\dnlim$ denotes weak convergence. We will also use stochastic ordering notations throughout the proof. For a sequence of random variable $\{X_n\}$ and a sequence of real numbers $\{a_n\}$, we say  $X_n = \bigo_P(a_n)$ if for all $\epsilon >0$ there exists a $M>0$, such that for all $n \in \N$, we have $\Prob\left(\left|\frac{X_n}{a_n}\right|> M\right) < \epsilon$. Also, we say $X_n = \smallo_P(a_n)$ if for all $\epsilon > 0$, $\Prob\left(\left|\frac{X_n}{a_n}\right|> \epsilon\right) \nlim 0$. For a real valued function $f$ defined on an inner-product space $(\mathcal M,\langle\cdot,\cdot\rangle)$, we denote $\|f(\cdot)\|^2_{\mathcal M} \coloneqq \sup_{t \in \mathcal M} \langle f(t),f(t)\rangle$. All the integrals are over the support $\calD$ of the sampling density $f_S$, and hence unless specified, by $\int$ we will mean $\int_\calD$. Recall that $\calF(\calD,f_S)$ is a space of continuous functions in $\calD$ equipped with the inner product $\langle g,h\rangle_{f_S}\coloneqq \int_\calD g(s)h(s)f_S(s)ds$ and the corresponding norm $\|\cdot \|_{f_S}$ with respect to the sampling measure $f_S$.
\section{Supporting Lemmas}\label{supsec:supp-lemmas}
\begin{lemma}[Argmax Theorem \citep{vaart1997weak}]\label{lem:argmax-thm}
        Let $\bM_n$ be a stochastic process indexed by a metric space $\Theta$, and let $\bM: \Theta \mapsto \R$ be a deterministic function. Suppose that $\left\|\bM_n - \bM\right\|_\Theta$ in outer probability and that there exists a $\theta_0$ such that 
        $$
        \bM(\theta_0) > \underset{\theta \notin G}{\sup} \,\bM(\theta)
        $$
        for any open set $G$ that contains $\theta_0$. Then any sequence $\hat{\theta}_n$, such that $\bM_n(\hat{\theta}_n) \geq \sup_\theta \bM_n(\theta) - \smallo_P(1)$, satisfies $\hat{\theta}_n \to \theta$ in outer probability.
    \end{lemma}
\begin{lemma}\label{lem:dot-prod-lim}
    For any basis function $h_j$ from the orthogonal class $\calH$, defining its evaluation at coordinates $S_1, S_2,\cdots,S_n \overset{iid}{\sim} f_S>0$ as $\bh_{n,j} = (h_j(S_1), \cdots,h_j(S_n))^\top$ and for iid noise vector $\bzet_n = (\zeta_1, \cdots, \zeta_n) \sim N(\bzero,\sigma^2_\zeta\bI_n)$ where $\bzet_n \perp \bS_n = (S_1, \cdots, S_n)^\top$, the following holds:
    \begin{enumerate}
        \item\label{basis-basis-dot} $\dfrac{1}{n} \bh_{n,j}^\top\bh_{n,l} \pnlim \mathbb{I}(j = l)$ where $\mathbb{I}(\cdot)$ is the indicator function for any $j,l \geq 1$
        \item\label{basis-noise-dot} $\dfrac{1}{n} \bh_{n,j}^\top\bzet_n \pnlim 0$ for any $j \geq 1$
    \end{enumerate}
\end{lemma}
\begin{proof}
    Consider the fact that the basis function set $\calH = \left\{ h_k(\cdot) : \int_{\mathcal{D}} h_i(s) h_j(s) f_S(s) ds = \delta_{ij} \right\}$ where $\delta_{ij}=1$ iff $i=j$. Now, this Lemma follows trivially from the Weak Law of Large Numbers (WLLN). To show \ref{basis-basis-dot}, notice that:
    $$
    \begin{aligned}
        \dfrac{1}{n} \bh_{n,j}^\top\bh_{n,l} = & \dfrac{1}{n} \sum_{i=1}^n h_j(s_i)h_l(s_i) \pnlim \E h_j(S_1)h_l(S_1) = \int h_j(s)h_l(s)f_S(s)ds = \mathbb{I}(j=l)
    \end{aligned}
    $$
    Now, to see \ref{basis-noise-dot}, observe that:
    $$
    \E \dfrac{1}{n} \bh_{n,j}^\top\bzet_n= \E_{\bS_n}\dfrac{1}{n}\bh_{n,j}^\top\E_{\bzet_n|\bS_n}\left[\bzet_n|\bS_n\right] = \E_{\bS_n}\dfrac{1}{n}\bh_{n,j}^\top\E_{\bzet_n|\bS_n}\left[\bzet_n\right] = \bzero 
    $$
    Also, observe that:
    $$
    \begin{aligned}
        Var\left(\bh_{n,j}^\top\bzet_n\right) = & Var_{\bS_n}\E[\bh_{n,j}^\top\bzet_n|\bS_n] + \E_{\bS_n}Var[\bh_{n,j}^\top\bzet_n|\bS_n]\\
        = & \E_{\bS_n}Var[\bh_{n,j}^\top\bzet_n|\bS_n]\\
        = & \sigma^2_{\zeta}\E_{\bS_n}\bh_{n,j}^\top\bh_{n,j} = n\sigma^2_\zeta\E h_j^2(S_1)\\
        = & n\sigma^2_\zeta\int_{\calD}h^2_j(s)f_S(s)ds = n\sigma^2_\zeta\\
        \implies Var\left(\dfrac{1}{n}\bh_{n,j}^\top\bzet_n\right) =  & \dfrac{\sigma^2_\zeta}{n} \nlim 0
    \end{aligned}
    $$
    This shows that $\dfrac{1}{n}\bh_{n,j}^\top\bzet_n \pnlim 0$ and this completes the proof of Lemma \ref{lem:dot-prod-lim}.
\end{proof}
\begin{lemma}\label{prop:tight-matrix-inv}
    Let $\bA_n$ be a $k \times k$ finite dimensional matrix of the form $\bA_n = \frac{1}{n}\sum_{l=1}^n A_\ell$ such that each of $A_\ell$ are iid square matrices of dimension $k$ such that $\E A_\ell = \bA$ with $||\bA||_2<\infty$. Then we have:
$$
\left|\left|\sqrt{n}\left(\bA_n^{-1} - \bA^{-1}\right)\right|\right|_2 = \bigo_P(1)
$$
where $||\cdot||_2$ is the spectral norm.
\end{lemma}
\begin{proof}
From Central Limit Theorem for matrices, we can say that the scaled difference $\sqrt{n}\left(\bA_n - \bA\right) \dnlim \bZ$ where $\bZ$ is a matrix normal random variable with mean matrix to be a square $\bzero$ matrix of dimension $k$ and covariance matrix same as the covariance of $vec(A_l)$, and since $||\cdot||_2$ is a continuous function on matrices, thus we can say that $||\sqrt{n}(\bA_n - \bA)||_2 = \bigo_P(1)$. Defining $\bDel_n = \bA_n - \bA$, we have:
$$
\begin{aligned}
    \bA_n^{-1} = & (\bA_n + \bA - \bA)^{-1}\\
    = & (\bA + \bDel_n)^{-1}\\
    = & \bA^{-1}\left(\bI + \bDel_n\bA^{-1}\right)^{-1}\\
    = & \bA^{-1}\sum_{k\geq 0}\left(-\bDel_n\bA^{-1}\right)^k\\
    = & \bA^{-1} -\bA^{-1}\bDel_n\bA^{-1} + \bA^{-1}\left(-\bDel_n\bA^{-1}\right)^2 + \bA^{-1}\left(-\bDel_n\bA^{-1}\right)^3 + \cdots \\
    = & \bA^{-1} -\bA^{-1}\bDel_n\bA^{-1} + \bigo\left(||\bDel_n||_2^2\right)
\end{aligned}
$$
Now, as $||\sqrt{n}\bDel_n||_2 = \bigo_P(1)$, it implies that $\sqrt{n}||\bDel_n||_2^2 = \bigo_P(n^{-1/2}) = \smallo_P(1)$. Thus, we have:
$$
\left|\left|\sqrt{n}(\bA_n^{-1} - \bA^{-1})\right|\right|_2 \leq ||\bA^{-1}||_2^2||\sqrt{n}\bDel_n||_2 + \bigo_P(\sqrt{n}||\bDel_n||_2^2) = \bigo_P(1) 
$$
This completes the proof of the proposition. \end{proof}
\begin{lemma}\label{lem:plurality-guarantee}
If $K$ is Lipschitz continuous and the finite set $\ba_m = \{a_1, \cdots, a_m\}$ has a unique mode $t$, then for all $0<h \leq h_0$, where $h_0$ is some constant depending on $K$ and $\underset{a_i \neq t}{\min} |t - a_i|$, $
    \underset{x\in \R}{\argmax}\ KD(\ba_m, K,h,x)=t
    $.
\end{lemma}
\begin{proof}
Recall we denoted the set $\ba_m= \{a_1,\cdots,a_m\}$ where $m \in \N$ and each of $a_i \in \R$. By the assumption of the lemma, the elements of $\ba_m$ have a unique mode at $t$. Formally, if we define $n_a$ to be the count of the elements of $\ba_m$ equaling the value $a$, then the unique mode assumption can be boiled down to $n_t > \underset{s \in \ba_m, s \neq t}{\max}n_s$. Let $\delta\coloneqq\underset{s\in \ba_m\setminus\{t\}}{\min} |t-s|>0$ and $M\coloneqq\underset{s\in \ba_m\setminus\{t\}}{\max} n_s$. Let us define the \emph{multiplicity gap} $\gamma\coloneqq n_t-M>0$. We consider the two assumptions on $K$ separately.

\medskip
\noindent\textbf{Case (i) [$K$ has compact support]:}
Suppose ${\rm supp}(K)\subseteq[-R,R]$ for some $R<\infty$. Choose $h_0\in\bigl(0,\delta/(2R)\bigr)$. Then for any $0<h\le h_0$, each summand $K\bigl((x-a_i)/h\bigr)$ is supported on an interval of radius $Rh\le \delta/2$. Hence the $n_s$ bumps centered at a location $s\in \ba_m$ do not overlap with those centered at any $s'\neq s$. In particular, for $s\in S$, $KD(\ba_m,K,h,s)=\frac{n_s K(0)}{mh}$,
since only the $n_s$ bumps centered at $s$ contribute at $x=s$, each attaining its maximum $K(0)$; and for some $s'\neq s$, we have $|s-s'|/h> \delta/h > 2R$ which throws the value outside the support of $K$, and we know for $x$ outside the support of all bumps, $K(x)=0$. Therefore,
$$
\max_{x\in\mathbb{R}} KD(\ba_m,K,h,x)
= \max_{s\in \ba_m} KD(\ba_m,K,h,s)
= \frac{K(0)}{mh}\,\max_{s\in \ba_m} n_s
= KD(\ba_m,K,h,t),
$$
and the uniqueness follows from $n_t>M$.

\medskip
\noindent\textbf{Case (ii)[$K(t)\to 0$ as $|t|\to\infty$]:}
By symmetry and continuity, $0\leq K(t)\leq K(0)$ for all $t$, and for every $\varepsilon>0$ there exists $T_\varepsilon<\infty$ such that  $\sup_{|t|\geq T_\varepsilon} K(t)\leq \varepsilon$. Fix small $\varepsilon>0$ so that
\begin{equation}\label{eq:eps-choice}
\varepsilon \;<\; \frac{K(0)}{2m}\gamma.
\end{equation}
Choose $h_0>0$ with $\delta/(2h_0)\ge T_\varepsilon$. For any $0<h\leq h_0$ and $x\in\mathbb{R}$, let $s^\star(x)\in \ba_m$ be a nearest element of $\ba_m$ to $x$ (breaking ties arbitrarily). Points $a_i\neq s^\star(x)$ satisfy $|x-a_i|\geq \delta/2$, hence
$$
K\left(\frac{x-a_i}{h}\right) \;\le\; \sup_{|u|\geq \delta/(2h)} K(u) \leq \varepsilon,
$$
and thus
\begin{equation}\label{eq:upper-anyx}
KD(\ba_m,K,h,x)
\le \frac{n_{s^\star(x)}K(0)}{mh} \;+\; \frac{(m-n_{s^\star(x)})\varepsilon}{mh}.
\end{equation}
Evaluating at $x=t$ gives
\begin{equation}\label{eq:lower-atm}
KD(\ba_m,K,h,t)
\ge \frac{n_t K(0)}{mh} - \frac{(m-n_t)\varepsilon}{mh}.
\end{equation}
Combining \eqref{eq:upper-anyx} and \eqref{eq:lower-atm} with $M=\max_{s\in \ba_m\neq t} n_s$, $\gamma=n_t-M$ and defining $\mathcal{B}(t,r) = \{x \in \R: |x-t|<r\}$, we have for any $x\in\R \setminus \mathcal{B}(t,\delta/2)$:
\begin{equation}\label{eq:KD-bound-ball-out}
    \begin{split}
& KD(\ba_m,K,h,t) - KD(\ba_m,K,h,x)\\
\geq &\frac{1}{mh}\Bigl(n_t K(0)-(m-n_t)\varepsilon - MK(0) - (m - M)\varepsilon\Bigr)\\
= & \frac{1}{mh}\Bigl((n_t -M)K(0)-(m + m-n_t-M)\varepsilon\Bigr)\\
\geq & \frac{1}{mh}\Bigl((n_t -M)K(0)-2m\varepsilon\Bigr)\\
= & \frac{2}{h}\Bigl(\gamma \frac{K(0)}{2m}-\varepsilon\Bigr) > 0
    \end{split}
\end{equation}
By the choice \eqref{eq:eps-choice}, the right-hand side is strictly positive. 

Now consider the the fixed annulus $A \coloneqq \{x:0<|x-t|\le\delta/2\} = \mathcal B(t,\delta/2)\setminus \{t\}$.
Define $\eta(x,h)\ \coloneqq K(0)-K((x-t)/h)$ where $x \in A$. Since $K$ attains its maximum only at $x=0$, we have $\eta(x,h)>0$ for every $x \in A$. moreover $\eta(x,h)\uparrow K(0)$ as $h\downarrow 0$.
Thus, decreasing $h_0$ if needed, ensure that for all $x \in A$,
\begin{equation}\label{eq:bound-KD}
    n_t\cdot\eta(x,h)\ \ge\ 3(m-n_t)\,\varepsilon\ ,\forall\ 0<h\le h_0.
\end{equation}
hence for any $x\in A$ and $0<h\le h_0$,
\begin{equation}\label{eq:bound-KD-diff-ball-in}
    KD(\ba_m,K,h,t)-KD(\ba_m,K,h,x)
\;\ge\; \frac{1}{mh}\Big(n_t \cdot \eta(x,h) - 2(m-n_t)\varepsilon\Big)\;>\;0.
\end{equation}
Now fix any $r>0$, since
$$
\R\setminus\mathcal B(t,r)\ \subset\
\big(\R\setminus\mathcal B(t,\delta/2)\big)\ \cup\ \{x:\ r\le |x-t|\le \delta/2\}\ \subset\
\big(\R\setminus\mathcal B(t,\delta/2)\big)\ \cup\ A,
$$
Combining Equations \ref{eq:KD-bound-ball-out}--\ref{eq:bound-KD-diff-ball-in} yields that for any $0<h\le h_0$ and $x\in \R \setminus \mathcal B(t,r)$, we have:
$$
KD(\ba_m,K,h,t) > KD(\ba_m,K,h,x),
$$
Hence $t$ is the unique global mode in this setup. 
\end{proof}
\begin{lemma}\label{lem:non-zero-proj}
Let $\{h_k\}_{k\ge1}\subset \calF(\calD,f_S)$ an orthonormal system, fix $d\in\mathbb N$, assume that for each $j\le d$, the zero set $\{s: h_j(s)=0\}$ has $f_S$-measure $0$
(equivalently, $h_j^2>0$ $f_S$-a.e.).
For $j\le d$, set $V_j\coloneqq\operatorname{span}\{h_k:\,1\le k\le d,\ k\neq j\}$ and define the set $\mathbb S_j\coloneqq\{v\in V_j:\ \|v\|_{f_S}=1\}$ and $Q_j(v)\coloneqq\int h_j^2\,v^2\,f_S$. Then
$$
c_\ast \coloneqq \inf_{1\le j\le d}\ \inf_{v\in \mathbb S_j} Q_j(v) > 0.
$$
\end{lemma}
\begin{proof}
\emph{Step 1 (strict positivity).} Fix $j\le d$ and $v\in V_j\setminus\{0\}$ such that $v \in \mathbb S_j$. Then $A\coloneqq\{s:\,v(s)\neq 0\}$ has positive $f_S$-measure. Since $h_j^2>0$ $f_S$-a.e., $B\coloneqq\{s:\,h_j(s)\neq 0\}$ has full $f_S$-measure,
and hence $f_S(A\cap B)=f_S(A)>0$.
On $A\cap B$ we have $h_j^2v^2>0$, so
$$
Q_j(v)=\int h_j^2v^2\,f_S>0.
$$
Thus $Q_j$ is strictly positive on $\mathbb S_j$.

\emph{Step 2 (compactness and continuity under orthonormal bases).} Let $\{\phi_1,\ldots,\phi_r\}$ be any orthonormal basis of $V_j$ in $\calL^4(\calD,f_S)$, where $r=d-1$. For $a=(a_1,\ldots,a_r)\in\mathbb R^r$ define $T(a)=\sum_{i=1}^r a_i\phi_i$. Because the basis is orthonormal, $T$ is an \emph{isometry} i.e. $\|T(a)\|_{f_S} = \|a\|_2$. Hence $\mathbb S_j = \{T(a):\ \|a\|_2=1\} = T(S^{r-1})$, where $S^{r-1}=\{a\in\mathbb R^r:\|a\|_2=1\}$ is the Euclidean unit sphere. Since $S^{r-1}$ is compact and $T$ is continuous, $\mathbb S_j$ is compact. Moreover,
$$
Q_j(T(a)) = \sum_{i,k=1}^r a_i a_k
\underbrace{\int h_j^2\,\phi_i\,\phi_k\,f_S}_{\boldsymbol{G}^{(j)}_{i,k}}
= a^\top \boldsymbol{G}^{(j)} a,
$$
where $G^{(j)}$ is a symmetric $r\times r$ matrix.
Thus $Q_j$ is continuous on $V_j$.

\emph{Step 3 (attainment of a strictly positive minimum).}
By the extreme value theorem, the continuous function $Q_j$ attains a minimum
on the compact set $\mathbb S_j$; denote this minimum by $c_j$.
From \textit{Step 1}, $Q_j(v)>0$ for every $v\in\mathbb S_j$, so $c_j>0$.
Since there are finitely many indices $1 \leq j \leq d$, set $c_\ast=\min_{1\leq j\leq d} c_j>0$.
\end{proof}
\begin{lemma}\label{lem:truncated-variance}
Let $\calH = \{h_j\}_{j\geq 1}$ be orthonormal in $\calF(\calD,f_S)$, with
$U=\sum_{j\geq 1}\alpha_u^{(j)} h_j$, $U_{\leq d}=\sum_{j\leq d}\alpha_u^{(j)} h_j$,
$U_{>d}=\sum_{j>d}\alpha_u^{(j)} h_j$ for $d\in\mathbb N$ and define
$S_{\mathrm{in}}\coloneqq \sum_{j\leq d}{\alpha_u^{(j)}}^2$.
Assume:
\begin{enumerate}[(i)]
\item\label{asm:infty-bound-basis} $\max_{1\le j\le d}\|h_j\|_{\infty}\leq M<\infty$.
\item\label{asm:postive-square-basis} For each $j\leq d$, the zero set $\{s: h_j(s)=0\}$ has $f_S$-measure $0$
(equivalently, $h_j^2>0$ $f_S$-a.e.).
\item\label{asm:c-star} $\|U_{>d}\|_{f_S}^2\le \varepsilon$ with
$\displaystyle \varepsilon \;<\; \frac{c_\ast}{4M^2}\,S_{\mathrm{in}}$, where for $V_j\coloneqq\mathrm{span}\{h_k: k \leq d,k\neq j\}$,
$$
c_\ast \coloneqq \min_{1\le j\le d}\ \inf_{\substack{v\in V_j\\ \|v\|_{f_S}=1}}
\int_\calD h_j^2\,v^2\,f_S\,ds > 0
$$
\end{enumerate}
Then for every $j=1,\dots,d$ with $j \in O_{X,d}$ i.e. $\alpha_u^{(j)}=0$,
$$
\|h_j U\|^2_{f_S} > \|h_j U_{>d}\|^2_{f_S}.
$$
\end{lemma}
\begin{proof}
We refer to Lemma \ref{lem:non-zero-proj} for this proof. Fix $1 \leq j\leq d$ and assume $\alpha_u^{(j)}=0$ i.e. $j \in O_{X,d}$. Write $R_j\coloneqq\sum_{k\leq d,\,k\neq j}\alpha_u^{(j)} h_k$ and $T\coloneqq U_{>d}$.
Set
$
A_j\coloneqq\int h_j^2 R_j^2\,f_S\,ds,
B_j\coloneqq\int h_j^2 T^2\,f_S\,ds$ and $
C_j\coloneqq\int h_j^2 R_j T\,f_S\,ds.
$
Then
$$
\|h_jU\|_{f_S}^2-\|h_jU_{>d}\|_{f_S}^2
=\int h_j^2\,(R_j^2+2R_jT)\,f_S\,ds
=A_j+2C_j.
$$
By Cauchy-Schwarz, $|C_j|\le\sqrt{A_j B_j}$, hence
$$
A_j+2C_j \;\ge\; A_j - 2\sqrt{A_jB_j}
\;=\; \sqrt{A_j}\,\big(\sqrt{A_j}-2\sqrt{B_j}\big).
$$
We will show that under the assumptions of the lemma, $\sqrt{A_j}>2\sqrt{B_j}$. First, by Assumption (\ref{asm:postive-square-basis}) which ensures $c_\ast >0$ (See Lemma \ref{lem:non-zero-proj}) and $\alpha_u^{(j)}=0$,
$$
A_j = \int h_j^2 R_j^2\,f_S\,ds
\geq c_\ast\,\|R_j\|_{f_S}^2
= c_\ast\sum_{k\le d,k\neq j}{\alpha_u^{(k)}}^2
= c_*\,S_{\mathrm{in}}.
$$
Second, by Assumptions (\ref{asm:infty-bound-basis}) and (\ref{asm:c-star}),
$$
B_j = \int h_j^2 T^2\,f_S\,ds
\leq \|h_j\|_\infty^2\cdot\|T\|_{f_S}^2
\;\le\; M^2\,\varepsilon.
$$
Thus,
$$
\sqrt{A_j}\ \ge\ \sqrt{c_*\,S_{\mathrm{in}}}
\;>\; 2\sqrt{M^2\,\varepsilon}\ \geq 2\sqrt{B_j},
$$
where the strict inequality is exactly assumption (\ref{asm:c-star}).
Therefore $A_j+2C_j>0$, i.e.,
$\|h_jU\|_{f_S}^2>\|h_jU_{>d}\|_{f_S}^2$.
\end{proof}
\section{Proofs of Theorems}
\subsection{Proof of Theorem \ref{thm:mode-consistent}}
\begin{proof}
        For the proof of this theorem, we will use Lemma \ref{lem:argmax-thm} on the convergence of argmax of stochastic processes. First we will show that the supremum norm of the difference $\left|\left|KD(\calC_{n,d},K,h,\cdot) - KD(\beta+R_d,K,h,\cdot)\right|\right|_\R\pnlim 0$ and then invoke the lemma to prove consistency of our mode estimator. Observe that since the kernel $K()$ is assumed to be Lipschitz continuous, then for any $x,y \in \R$ and a constant $L_K > 0$ (dependent on $K$), we can say $|K(x) - K(y)| \leq L_K |x-y|$. Thus, using this, we have:
    $$
    \begin{aligned}
        & \|KD(\calC_{n,d},K,h,x) - KD(\beta+R_d,X,h,x)\|_\R\\ 
        = & \underset{x \in \R}{\sup}\left|KD(\calC_{n,d},K,x) - KD(\beta+R_d,K,x)\right|\\
        = & \underset{x \in \R}{\sup}\left|\dfrac{1}{mh}\sum_{k \in S_{X,d}}K\left(\dfrac{x - \hat{\beta}_{k,n}}{h}\right) - \dfrac{1}{mh}\sum_{k \in S_{X,d}}K\left(\dfrac{x - (\beta + r_k)}{h}\right)\right|\\
        \leq & \underset{x \in \R}{\sup}\dfrac{1}{mh}\sum_{k \in S_{X,d}}\left|K\left(\dfrac{x - \hat{\beta}_{k,n}}{h}\right) - K\left(\dfrac{x - (\beta + r_k)}{h}\right)\right|\\
        \leq & \underset{x \in \R}{\sup}\dfrac{L_K}{mh}\sum_{k \in S_{X,d}}\left|\left(\dfrac{x - \hat{\beta}_{k,n}}{h}\right) - \left(\dfrac{x - (\beta + r_k)}{h}\right)\right| \quad \text{(using Lipschitz Continuity)}\\
        = & \dfrac{L_K}{mh^2}\sum_{k \in S_{X,d}}\left|\hat{\beta}_{k,n}-(\beta + r_k)\right| \pnlim 0 
    \end{aligned}
    $$
    This last step is because $\frac{L_K}{mh^2}< \infty$ and $|S_{X,d}| = m$, which is fixed, and each of the summands goes to 0 in probability as per the assumption of the theorem. Hence, we have the desired result for the sup-norm of the difference of the $KD$'s. 

    Now, based on the Plurality Rule Assumption \ref{asm:estimate} for some $d \geq d_0$, and Lemma \ref{lem:plurality-guarantee},   
    it is ensured that for $\theta_0 = \beta$, we have $KD(\beta,K,h,\theta_0) > \underset{\theta \notin G}{\sup} KD(\theta+R_d,K,h,\theta)$ for any open set $G$ that contains $\theta_0$. 
    
    Thus defining $\hat{\beta}_n = \underset{x}{\argmax} 
    \ KD(\calC_{n,d},K,h,x)$, we can say that $KD\left(\calC_{n,d},K,h,\hat{\beta}_n\right) \geq \sup_\theta KD\left(\calC_{n,d},K,h,\theta\right) - \smallo_P(1)$. Thus, using Lemma \ref{lem:argmax-thm}, we can say that $\hat{\beta}_n \pnlim \beta$. This completes the proof of the theorem.
\end{proof}

\subsection{Proof of Theorem \ref{thm:IV-est-prop}}\label{thm:IV-est-prop-proof}
\begin{proof}
Let us choose some $j \in \N$, such that $\alpha_x^{(j)}\neq 0$. Now, let us look at $\hat{\beta}_{j,n}^{PROJ}$:
    $$
    \begin{aligned}
        \hat{\beta}_{j,n}^{PROJ}  = \dfrac{\dfrac{1}{n}\bh_{n,j}^\top\bY_{n}}{\dfrac{1}{n}\bh_{n,j}^\top\bX_{n}} = & \dfrac{\dfrac{1}{n}\bh_{n,j}^\top(\beta\bX_{n}+\bU_{n} + \beps_{n})}{\dfrac{1}{n}\bh_{n,j}^\top\bX_{n}}\\
        = & \beta + \dfrac{\dfrac{1}{n}\bh_{n,j}^\top\bU_{n}}{\dfrac{1}{n}\bh_{n,j}^\top\bX_{n}} + \dfrac{\dfrac{1}{n}\bh_{n,j}^\top\beps_{n}}{\dfrac{1}{n}\bh_{n,j}^\top\bX_{n}}
    \end{aligned}
    $$
where we have $\bX_{n} = \sum_{k=1}^\infty\alpha_x^{(k)}\bh_{n,k}$ and $\bU_n = \sum_{k=1}^\infty\alpha_u^{(k)}\bh_{n,k}$. By WLLN, $\frac{1}{n}\bh_{n,j}^\top\bX_n \pnlim \E h_j(S)X(S)$, where we will now try to evaluate the expectation. As we know that $X(s) = \sum_{j=1}^\infty\alpha_{x}^{(j)}h_j(s)$, let us define $X_{k}(s) = \sum_{j=1}^k\alpha_{x}^{(j)}h_j(s)$ that is part of $X$ explained by the first $k$ basis functions of $\calH$. Then we have $h_j(s)X_{k}(s) \klim h_j(s)X(s)$ pointwise for every $s \in \calD$ and we have:
$$
\begin{aligned}
\E h_j(S)X_{k}(S) =&\int h_j(s)X_{k}(s)f_Sds\\
= & \sum_{u=1}^k\alpha_x^{(u)}\int h_j(s)h_u(s)\,f_S\,ds\\
= & \sum_{u=1}^k\alpha_x^{(u)}\mathbb{I}(u=j)\\
\leq & \sum_{u=1}^k\left|\alpha_x^{(u)}\right|\mathbb{I}(u=j)\leq \left(\sum_{u \geq 1} \alpha_x^{(u)^2}\right)^{1/2}< \infty
\end{aligned}
$$
Thus Dominated Convergence Theorem (DCT) implies that:
$$
\E h_j(S)X_{k}(S) \klim \E h_j(S)X(S) = \sum_{u=1}^\infty\alpha_x^{(u)}\mathbb{I}(u=j) = \alpha_x^{(j)}
$$
Similarly we have $\frac{1}{n}\bh_{n,j}^\top\bU_n \pnlim \alpha_u^{(j)}$. Lemma \ref{lem:dot-prod-lim} implies that $\frac{1}{n}\bh_{n,j}^\top\beps_n \pnlim 0$. Thus when $\alpha_x^{(j)} \neq 0$, we have $\hat{\beta}_{j,n}^{PROJ} \pnlim \beta + \alpha_u^{(j)}/\alpha_x^{(j)} = \beta + r_j$ which completes the proof of consistency.
Let us define $A_i = h_j(S_i)[U(S_i)+\epsilon_i]$ and $B_n = h_j(S_i)X(S_i)$,which means $\frac{1}{n}\bh_{n,j}^\top[\bU_n+\beps_n] = \frac{1}{n}\sum_{i=1}^nh_j(S_i)[U(S_i)+\epsilon_i] = \bar{A}_n$ and $\frac{1}{n}\bh_{n,j}^\top\bX_n = \frac{1}{n}\sum_{i=1}^nh_j(S_i)X(S_i) = \bar{B}_n$. We will try to derive the asymptotic distribution of $\bar{A}_n/\bar{B}_n$, which is the bias of the estimator $\hat{\beta}^{PROJ}_{j,n}$. Now we have seen that 
        $\E A_i =\E h_j(S_i)[U(S_i)+\epsilon_i] = \E h_j(S_i)U(S_i) = \int h_j(s)U(s)f_S(s)ds = \alpha_u^{(j)}$ and $\E B_i =\E h_j(S_i)X(S_i) = \alpha_x^{(j)}$
Also, notice that 
$$
\begin{aligned}
    \E A_i^2 & = \E \left[h_j^2(S_i)(U(S_i) + \epsilon_i)^2\right]\\
    & = \E [h_j(S_i)U(S_i)]^2 + \E \left[h^2_j(S_i)\epsilon_i^2\right] + 2\E [h_j(S_i)U(S_i)\epsilon_i]\\
    & = \int [h_j(s)U(s)]^2f_S(s)\,ds +  \sigma^2_\epsilon\\
    & = \|h_jU\|^2_{f_S} + \sigma^2_\epsilon
\end{aligned}
$$   
And similarly, we will have $\E B_i^2 = \int [h_j(s)X(s)]^2f_S(s)\,ds = \|h_jX\|^2_{f_S}$. Also, note that:
$$
\begin{aligned}
    \E A_i B_i = & \E [h_j(S)(U(S) + \epsilon_i)h_j(S)X(S) ]\\
    = & \int h_j^2(s)U(s)X(s)f_S(s)\,ds 
\end{aligned}
$$
Since $\int(h_j U(s))^2f_Sds< \infty$ and $\int(h_j X(s))^2f_Sds < \infty $ as the functions are elements of $\calF(\calD,f_S)$, by Multivariate Central Limit Theorem (CLT), we have:
    \begin{equation}\label{eq:mult-clt}
    \sqrt{n}\left\{\begin{pmatrix}
        \bar{A}_n\\\bar{B}_n
    \end{pmatrix} -  \begin{pmatrix}
        \alpha_u^{(j)}\\\alpha_x^{(j)}
    \end{pmatrix}\right\} = 
     \dfrac{1}{\sqrt{n}}\sum_{i=1}^n\left\{\begin{pmatrix}
        A_i\\B_i
    \end{pmatrix} -  \begin{pmatrix}
        \alpha_u^{(j)}\\\alpha_x^{(j)}
    \end{pmatrix}\right\} \dnlim N\left(\bzero_2,\bSig_{AB}\right)
    \end{equation}
where $\bzero_2 = (0,0)^\top$, and $\bSig_{AB} = \begin{pmatrix}
    \sigma^2_\epsilon + \sigma^2_{U,j} & \sigma_{U,X,j}\\
    \sigma_{U,X,j} & \sigma^2_{X,j} 
\end{pmatrix}$ with:
    $$
\begin{aligned}
    \sigma^2_{U,j} = & \int (h_j(s)U(s))^2f_S(s)ds   - {\alpha_u^{(j)}}^2 = \|h_j U\|^2_{f_S} - {\alpha_u^{(j)}}^2\\
    \sigma^2_{X,j} = & \int (h_j(s)X(s))^2f_S(s)ds   - {\alpha_x^{(j)}}^2 = \|h_j X\|^2_{f_S} - {\alpha_x^{(j)}}^2\\
        \sigma_{U,X,j} = & \int h^2_j(s)X(s)U(s)f_S(s)ds  - \alpha_x^{(j)}\alpha_u^{(j)} = \langle h_jX, h_jU\rangle_{f_S} - {\alpha_u^{(j)}}{\alpha_x^{(j)}}\\
\end{aligned}
$$
Now, using Delta Method on the bivariate function $g: (x,y) \mapsto x/y$, defining $\nabla g = \left(\dfrac{1}{\alpha_x^{(j)}}, -\dfrac{\alpha_u^{(j)}}{{\alpha_x^{(j)}}^2}\right)^\top$, we have:
\begin{equation}\label{eq:clt-ratio}
    \sqrt{n}\left(\dfrac{\bar{A}_n}{\bar{B}_n} - \underbrace{\dfrac{\alpha_u^{(j)}}{\alpha_x^{(j)}}}_{r_j}\right) \dnlim N\left(0, \sigma^2_{j}\right)
\end{equation}
where:
$$
\begin{aligned}
\sigma^2_{j} = \nabla g^\top\bSig_{AB}\nabla g = &\dfrac{\sigma^2_{U,j}+ \sigma^2_\epsilon}{{\alpha_x^{(j)}}^2} + \left(\dfrac{\alpha_u^{(j)}}{{\alpha_x^{(j)}}^2}\right)^2\sigma^2_{X,j} - 2\dfrac{\alpha_u^{(j)}}{{\alpha_x^{(j)}}^3}\sigma_{U,X,j}\\
= & \dfrac{\sigma^2_{U,j} + \sigma^2_\epsilon}{{\alpha_x^{(j)}}^2} + \alpha_u^{(j)}\left(\dfrac{\alpha_u^{(j)}}{{\alpha_x^{(j)}}^4}\sigma^2_{X,j} - \dfrac{2}{{\alpha_x^{(j)}}^3}\sigma_{U,X,j}\right)
\end{aligned}
$$
Thus, if $j \in O_X$, then $\alpha_u^{(j)} = 0$ and thus $r_j = 0$, which means $\sigma_{j}^2 = \frac{\sigma^2_{U,j} + \sigma^2_\epsilon}{\alpha_x^{(j)^2}} $. Thus for such a $j\in O_X$, we have:
$$
\sqrt{n}\left(\hat{\beta}^{PROJ}_{j,n} - \beta\right) \dnlim N\left(0, \dfrac{\sigma^2_{U,j} + \sigma^2_\epsilon}{\alpha_x^{(j)^2}}\right)
$$
This completes the proof of Theorem \ref{thm:IV-est-prop} and also proves Corollary \ref{cor:IV-consistent}.
\end{proof}

\subsection{Proof of Theorem \ref{thm:drop-est-prop}}\label{thm:drop-est-prop-proof}
\begin{proof}
    Recall the notations that we will be using for the proof:
    $$
    \begin{aligned}
        \bH_{1:d} = &\begin{bmatrix}
    \bh_{n,1} & \cdots & \bh_{n,d}
\end{bmatrix} \\
        \bH_{-j|1:d} = &
\begin{bmatrix}
    \bh_{n,1} & \cdots & \bh_{n,j-1} & \bh_{n,j+1} & \cdots & \bh_{n,d}
\end{bmatrix} \\
\bP_{d} = & \bH_{1:d}\left(\bH_{1:d}^\top\bH_{1:d}\right)^{-1}\bH_{1:d}^\top\\
\bP_{-j,d} = & \bH_{-j|1:d}\left(\bH_{-j|1:d}^\top\bH_{-j|1:d}\right)^{-1}\bH_{-j|1:d}^\top\\
\bX_{n,d} = & \bP_d\bX_n\\
\bY_{n,d} = & \bP_d\bY_n
    \end{aligned}
    $$
    We will simplify certain notations $\bH_{1:d} \to \bH_d$ and $\bH_{-j|1:d} \to \bH_{-j,d}$.
    Using these notations, we define out drop-one estimator for $j \in S_{X,d}$ can be simplified as:
    $$
    \begin{aligned}
            \hat{\beta}_{-j,n} = & \dfrac{\bX_{n,d}^\top(\bI_n - \bP_{-j,d})\bY_{n,d}}{\bX_{n,d}^\top(\bI_n - \bP_{-j,d})\bX_{n,d}}\\
            = & \dfrac{\bX_{n}^\top\bP_d(\bI_n - \bP_{-j,d})\bP_d\bY_{n}}{\bX_{n}^\top\bP_d(\bI_n - \bP_{-j,d})\bP_d\bX_{n}}\\
            = & \dfrac{\bX_{n}^\top\bP_d(\bI_n - \bP_{-j,d})\bP_d(\bX_n\beta + \bU_n + \beps_n)}{\bX_{n}^\top\bP_d(\bI_n - \bP_{-j,d})\bP_d\bX_{n}}\\
            = & \beta + \dfrac{\bX_{n}^\top\bP_d(\bI_n - \bP_{-j,d})\bP_d(\bU_n + \beps_n)}{\bX_{n}^\top\bP_d(\bI_n - \bP_{-j,d})\bP_d\bX_{n}}\\
    \end{aligned}
    $$
Notice that:
\begin{equation}\label{eq:proj-simplify}
    \begin{split}
        & \frac{1}{n}\bP_d(\bI_n - \bP_{-j,d})\bP_d\\
        = & \frac{1}{n}\bH_{d}\left(\bH_{d}^\top\bH_{d}\right)^{-1}\bH_{d}^\top(\bI_n - \bP_{-j,d})\bH_{d}\left(\bH_{d}^\top\bH_{d}\right)^{-1}\bH_{d}^\top\\
        = & \frac{1}{n}\bH_{d}\left(\bH_{d}^\top\bH_{d}\right)^{-1}\bH_{d}^\top\left(\bI_n - \bH_{-j,d}\left(\bH_{-j,d}^\top\bH_{-j,d}\right)^{-1}\bH_{-j,d}^\top\right)\bH_{d}\left(\bH_{d}^\top\bH_{d}\right)^{-1}\bH_{d}^\top\\
        = & \frac{1}{n}\bH_{d}\underbrace{\left(\frac{1}{n}\bH_{d}^\top\bH_{d}\right)^{-1}\left(\frac{1}{n}\bH_{d}^\top\bH_{d} - \frac{1}{n}\bH_{d}^\top\bH_{-j,d}\left(\frac{1}{n}\bH_{-j,d}^\top\bH_{-j,d}\right)^{-1}\frac{1}{n}\bH_{-j,d}^\top\bH_{d}\right)\left(\frac{1}{n}\bH_{d}^\top\bH_{d}\right)^{-1}}_{\bA_{j,d}}\frac{1}{n}\bH_{d}^\top\\
        = &  \frac{1}{n}\bH_{d} \bA_{j,d}\frac{1}{n}\bH_{d}^\top
    \end{split}
\end{equation}
Next we prove a lemma about the probability limit of $\bA_{j,d}$:
\begin{lemma}\label{lem:A-lim}
    For the $d \times d$ matrix $\bA_{j,d}$ defined in Equation \ref{eq:proj-simplify}, we have:
    \begin{itemize}
        \item $\bA_{j,d} \pnlim \bA^{lim}_{j,d}\coloneqq diag(\left\{\mathbb{I}(u = j), 1 \leq u \leq d\right\})$
        \item $\left|\left|\sqrt{n}\left(\bA_{j,d} - \bA^{lim}_{j,d}\right)\right|\right|_2 = \bigo_P(1)$ where $||\cdot||_2$ denotes the spectral norm.
    \end{itemize}
\end{lemma}
\begin{proof}
    Using Lemma \ref{lem:dot-prod-lim}, we can say that $\frac{1}{n}\bH_{d}^\top\bH_{d} \pnlim \bI_d$ and $\frac{1}{n}\bH_{-j,d}^\top\bH_{-j,d} \pnlim \bI_{d-1}$, and similarly using the same lemma, we can say that:
    \begin{equation}\label{eq:cross-basis-inner}
        \begin{split}
              \frac{1}{n}\bH_{-j,d}^\top\bH_{d} = & \frac{1}{n}\begin{bmatrix}
            \bh^\top_{n,1}\\
            \vdots\\
            \bh^\top_{n,j-1}\\
            \bh^\top_{n,j+1}\\
            \vdots\\
            \bh^\top_{n,d}
        \end{bmatrix}
        \begin{bmatrix}
    \bh_{n,1} & \cdots & \bh_{n,j-1}&\bh_{n,j}&\bh_{n,j+1}& \cdots& \bh_{n,d}
\end{bmatrix}\\
= & \frac{1}{n}\begin{bmatrix}
    \bH^\top_{1:(j-1),n}\\
     \bH^\top_{(j+1):d,n}\\
\end{bmatrix}_{(d-1)\times n}\begin{bmatrix}
    \bH_{1:(j-1),n} & \bh_{n,j} & \bH_{(j+1):d,n}
\end{bmatrix}_{n \times d}\\
\pnlim & \begin{bmatrix}
    \bI_{j-1} & \bzero_{j-1} & \bzero_{(j-1) \times (d-j)}\\
    \bzero_{(d-j)\times (j-1)} & \bzero_{d-j} & \bI_{d-j} 
\end{bmatrix} = \bR_{j,d}
        \end{split}
    \end{equation}
with:
\begin{equation}\label{eq:drop-inner-drop}
    \bR_{j,d}^\top\bR_{j,d} = \begin{bmatrix}
    \bI_{j-1} & \bzero_{j-1} & \bzero_{(j-1) \times (d-j)}\\
    \bzero^\top_{j-1} & 0 & \bzero_{d-j}\\
    \bzero_{(d-j) \times (j-1)} & \bzero_{d-j} & \bI_{d-j}\\
\end{bmatrix}
\end{equation}
Hence, we have:
$$
\begin{aligned}
    \bA_{j,d} = & \left(\frac{1}{n}\bH_{d}^\top\bH_{d}\right)^{-1}\left(\frac{1}{n}\bH_{d}^\top\bH_{d} - \frac{1}{n}\bH_{d}^\top\bH_{-j,d}\left(\frac{1}{n}\bH_{-j,d}^\top\bH_{-j,d}\right)^{-1}\frac{1}{n}\bH_{-j,d}^\top\bH_{d}\right)\left(\frac{1}{n}\bH_{d}^\top\bH_{d}\right)^{-1}\\
    \pnlim & \bI^{-1}_d\left(\bI_d - R^\top_{j,d}\bI_{d-1}^{-1}R_{j,d}\right)\bI^{-1}_d\\
    = & \bI_d - \begin{bmatrix}
    \bI_{j-1} & \bzero_{j-1} & \bzero_{(j-1) \times (d-j)}\\
    \bzero^\top_{j-1} & 0 & \bzero_{d-j}\\
    \bzero_{(d-j) \times (j-1)} & \bzero_{d-j} & \bI_{d-j} \\
\end{bmatrix} = diag(\left\{\mathbb{I}(u = j), 1 \leq u \leq d\right\}) = \bA^{lim}_{j,d}
\end{aligned}
$$
This completes the proof of the first part lemma. Next, we move to proving the uniform tightness of a matrix deviation. 

We will be using Lemma \ref{prop:tight-matrix-inv} in proving the second part of our lemma. Recall how We defined the matrix $\bA_{j,d}$ as:
$$
A_{j,d} = \left(\frac{1}{n} \bH_d^\top \bH_d \right)^{-1} \left( \frac{1}{n} \bH_d^\top \bH_d - \frac{1}{n} \bH_d^\top \bH_{-j,d} \left( \frac{1}{n} \bH_{-j,d}^\top \bH_{-j,d} \right)^{-1} \frac{1}{n} \bH_{-j,d}^\top \bH_d \right) \left(\frac{1}{n} \bH_d^\top \bH_d \right)^{-1}.
$$
and its probability limit $\bA_{j,d}^{lim} = \text{diag}(\mathbb{I}(u = j), 1 \leq u \leq d)$ which is a $d \times d$ matrix with $1$ at position $(j,j)$ and $0$ elsewhere. Let us define:
$$
\begin{aligned}
    \bM_n \coloneqq & \frac{1}{n} \bH_d^\top \bH_d \pnlim \bI_d\\
\bC_n^\top \coloneqq & \frac{1}{n} \bH_d^\top \bH_{-j,d} \pnlim \bR^\top_{j,d} \text{ (Eqn. \ref{eq:cross-basis-inner})}\\
\bD_n \coloneqq &\frac{1}{n} \bH_{-j,d}^\top \bH_{-j,d} \pnlim \bI_{d-1}    
\end{aligned}
$$
Then,
$$
\bA_{j,d} = \bM_n^{-1} \left(\bM_n - \bC_n^\top \bD_n^{-1} \bC_n\right) \bM_n^{-1} = \bM_n^{-1} - \bM_n^{-1}\bC_n^\top \bD_n^{-1} \bC_n \bM_n^{-1}
$$
Also, observe that we can write down $\bA_{j,d}^{lim} = \bI_d - \bR^\top_{j,d}\bR_{j,d}$ where $\bR^\top_{j,d}\bR_{j,d}$ has all the diagonal entries to be 1 except the $j$-th diagonal element which is 0 as seen in Equation \ref{eq:drop-inner-drop}. We denote $\bR^*_{j,d}\coloneqq \bR^\top_{j,d}\bR_{j,d}$. Thus, we have:
\begin{equation}
    \begin{split}
        & \sqrt{n}\left(\bA_{j,d} - \bA^{lim}_{j,d}\right)\\
        = & \sqrt{n}\left[\left(\bM_n^{-1} - \bI_d\right) - \left(\bM_n^{-1}\bC_n^\top \bD_n^{-1} \bC_n \bM_n^{-1} - \bR^*_{j,d}\right)\right]\\
        = & \sqrt{n}\left(\bM_n^{-1} - \bI_d\right) - \sqrt{n}\left[\bM_n^{-1}(\bC_n^\top \bD_n^{-1} \bC_n - \bR^*_{j,d} + \bR^*_{j,d})\bM_n^{-1} - \bR^*_{j,d}\right]\\
        = & \!\begin{aligned}[t]
        & \sqrt{n}\left(\bM_n^{-1} - \bI_d\right)\\
        - & \bM_n^{-1}\sqrt{n}\left(\bC_n^\top \bD_n^{-1} \bC_n - \bR^*_{j,d}\right)\bM_n^{-1}\\
        - &\sqrt{n}\left(\bM_n^{-1}\bR^*_{j,d}\bM_n^{-1} - \bR^*_{j,d}\right) \\
    \end{aligned}\\
    = & \!\begin{aligned}[t]
        & \sqrt{n}\left(\bM_n^{-1} - \bI_d\right)\\
        - & \bM_n^{-1}\left[\bC_n^\top \sqrt{n}(\bD_n^{-1} - \bI_{d-1}) \bC_n + \sqrt{n}(\bC_n^\top\bC_n - \bR^*_{j,d})\right]\bM_n^{-1}\\
        - &\sqrt{n}\left[(\bM_n^{-1}-\bI_d + \bI_d)\bR^*_{j,d}(\bM_n^{-1}-\bI_d + \bI_d) - \bR^*_{j,d}\right] \\
    \end{aligned}\\
    = & \!\begin{aligned}[t]
        & \sqrt{n}\left(\bM_n^{-1} - \bI_d\right)\\
        - & \bM_n^{-1}\left[\bC_n^\top \sqrt{n}(\bD_n^{-1} - \bI_{d-1}) \bC_n + \sqrt{n}(\bC_n^\top\bC_n - \bR^*_{j,d})\right]\bM_n^{-1}\\
        - &\left[\sqrt{n}(\bM_n^{-1}-\bI_d)\bR^*_{j,d}(\bM_n^{-1}-\bI_d)+\sqrt{n}(\bM_n^{-1}-\bI_d)\bR^*_{j,d} + \bR^*_{j,d}\sqrt{n}(\bM_n^{-1}-\bI_d)\right] \\
    \end{aligned}
    \end{split}    
\end{equation}
Unless otherwise specified, we will refer to $||\cdot||$ as the spectral norm for matrices. Thus, using these simplifications, we have:
\begin{equation}
    \begin{split}
        & \left|\left|\sqrt{n}\left(\bA_{j,d} - \bA^{lim}_{j,d}\right)\right|\right|_2\\
        \leq & \!\begin{aligned}[t]
        & \left|\left|\sqrt{n}\left(\bM_n^{-1} - \bI^{-1}_d\right)\right|\right|_2\\
        + & \left|\left|\bM_n^{-1}\left[\bC_n^\top \sqrt{n}(\bD_n^{-1} - \bI^{-1}_{d-1}) \bC_n + \sqrt{n}(\bC_n^\top\bC_n - \bR^*_{j,d})\right]\bM_n^{-1}\right|\right|_2\\
        + &
        \left|\left|\left[\sqrt{n}(\bM_n^{-1}-\bI^{-1}_d)\bR^*_{j,d}(\bM_n^{-1}-\bI^{-1}_d)+\sqrt{n}(\bM_n^{-1}-\bI^{-1}_d)\bR^*_{j,d} + \bR^*_{j,d}\sqrt{n}(\bM_n^{-1}-\bI^{-1}_d)\right] \right|\right|_2\\
        \end{aligned}\\
        \leq & \!\begin{aligned}[t]
        & \left|\left|\sqrt{n}\left(\bM_n^{-1} - \bI^{-1}_d\right)\right|\right|_2\\
        + & \left|\left|\bM_n^{-1}\right|\right|^2_2\left(\left|\left|\bC_n^\top \sqrt{n}(\bD_n^{-1} - \bI^{-1}_{d-1}) \bC_n\right|\right|_2 + \left|\left|\sqrt{n}(\bC_n^\top\bC_n - \bR^*_{j,d})\right|\right|_2\right)\\
        + &
        \left|\left|\sqrt{n}(\bM_n^{-1}-\bI^{-1}_d)\bR^*_{j,d}(\bM_n^{-1}-\bI^{-1}_d)\right|\right|_2+\left|\left|\sqrt{n}(\bM_n^{-1}-\bI^{-1}_d)\bR^*_{j,d}\right|\right|_2\\ 
        + &\left|\left|\bR^*_{j,d}\sqrt{n}(\bM_n^{-1}-\bI^{-1}_d)\right|\right|_2\\
        \end{aligned}\\
              \leq & \!\begin{aligned}[t]
        & \left|\left|\sqrt{n}\left(\bM_n^{-1} - \bI^{-1}_d\right)\right|\right|_2\\
        + & \left|\left|\bM_n^{-1}\right|\right|^2_2\left(\left|\left|\bC_n \right|\right|^2_2\left|\left|\sqrt{n}(\bD_n^{-1} - \bI^{-1}_{d-1})\right|\right|_2 + \left|\left|\sqrt{n}(\bC_n^\top\bC_n - \bR^*_{j,d})\right|\right|_2\right)\\
        + &
        \left(\left|\left|\sqrt{n}(\bM_n^{-1}-\bI^{-1}_d)\right|\right|_2\left|\left|\bR^*_{j,d}\right|\right|_2\left|\left|(\bM_n^{-1}-\bI^{-1}_d)\right|\right|_2+\left|\left|\sqrt{n}(\bM_n^{-1}-\bI^{-1}_d)\right|\right|_2\left|\left|\bR^*_{j,d}\right|\right|_2\right)\\
        + & \left|\left|\bR^*_{j,d}\right|\right|_2\left|\left|\sqrt{n}(\bM_n^{-1}-\bI^{-1}_d)\right|\right|_2\\
        \end{aligned}\\
    \end{split}
\end{equation}
Let us talk about the spectral norms of some of the components of the upper bound of the quantity $\left|\left|\sqrt{n}\left(\bA_{j,d} - \bA^{lim}_{j,d}\right)\right|\right|_2$. As $\bM_n$ converges to $\bI_d$ and matrix inversion and spectral norm are continuous maps on the space of matrices, hence we can say $\left|\left|\bM_n^{-1}\right|\right|^2_2 = \bigo_P(1)$. Similarly, we can say the same thing about $\left|\left|\bC_n \right|\right|^2_2$ being tight. Also note that since $\sqrt{n}(\bC_n - \bR_{j,d})$ has a matrix CLT due to its structure as $\bC_n$ can be written as sum of $n$ iid matrices of dimension $d \times (d-1)$, combined with the fact that $\bA \mapsto \bA^\top\bA$ is a continuous map, we can say $\left|\left|\sqrt{n}(\bC_n^\top\bC_n - \bR^\top_{j,d}\bR_{j,d})\right|\right|_2$ is $\bigo_P(1)$. Also from Equation \ref{eq:drop-inner-drop}, it is also clear that $\left|\left|\bR^*_{j,d}\right|\right|_2 = 1$. The remaining terms of the upper bound can be said to be $\bigo_P(1)$ by using Proposition \ref{prop:tight-matrix-inv}. This completes the proof of the second part of the lemma.
\end{proof}
Denoting $\balp_{x,d} = \left(\alpha^{(1)}_{x}, \cdots, \alpha^{(d)}_{x}\right)^\top$ and $\balp_{u,d} = \left(\alpha^{(1)}_{u}, \cdots, \alpha^{(d)}_{u}\right)^\top$, using Lemma \ref{lem:dot-prod-lim} and WLLN, we have:
\begin{equation}\label{eq:X-Hk-inner}
    \begin{split}
           \frac{1}{n}\bX_n^\top\bH_d = & \begin{bmatrix}
        \frac{1}{n}\bX_n^\top\bh_{n,1} & \cdots & \frac{1}{n}\bX_n^\top\bh_{n,d}
    \end{bmatrix}_{1 \times d}\\
    \pnlim & \begin{bmatrix}
        \alpha_x^{(1)} & \cdots & \alpha_x^{(d)}
    \end{bmatrix} = \balp_{x,d}^\top
    \end{split}
\end{equation}
Similarly, we have:
$$
    \frac{1}{n}\bU_n^\top\bH_d \pnlim\balp_{u,d}^\top, \frac{1}{n}\beps_n^\top\bH_d \pnlim\bzero_{d}^\top
$$
Thus using the above statements and Lemma \ref{lem:A-lim}, we have:
\begin{equation}
    \begin{split}
        \hat{\beta}_{-j,n} - \beta = &\dfrac{\frac{1}{n}{\bX}_{n}^\top\bP_d(\bI_n - \bP_{-j,d})\bP_d(\bU_n + \beps_n)}{\frac{1}{n}{\bX}_{n}^\top\bP_d(\bI_n - \bP_{-j,d})\bP_d\bX_{n}}\\
        = & \dfrac{\frac{1}{n}{\bX}_{n}^\top\bH_{d} \bA_{j,d}\frac{1}{n}{\bH}_{d}^\top(\bU_n + \beps_n)}{\frac{1}{n}{\bX}_{n}^\top\bH_{d} \bA_{j,d}\frac{1}{n}{\bH}_{d}^\top\bX_{n}}\\
        \pnlim & \frac{\balp_{x,d}^\top A^{lim}_{j,d}(\balp_{u,d} + \bzero_d)}{{\balp}_{x,d}^\top A^{lim}_{j,d}\balp_{x,d}} = \frac{\alpha_x^{(j)}\alpha_u^{(j)}}{\alpha_x^{(j)^2}} = \frac{\alpha_u^{(j)}}{\alpha_x^{(j)}} = r_j
    \end{split}
\end{equation}
This completes the proof for the consistency of the bias of our drop-one estimator. 

Next, we will show the asymptotic distribution of the drop one estimator when the index $j$ of the dropped basis function lies in $O_{X,d}$, i.e. $\alpha_x^{(j)} \neq 0$ and $\alpha_u^{(j)} = 0$. Notice that $\bU_n$ can be decomposed as $\bU_n = \bU_{n,d} + (1 - \lambda_{d})\bU_{n,>d}$, where the above condition also means that $\bU_{n,d} \in \colsp(\bH_{-j,d})$. Thus, we have:
\begin{equation}\label{eq:projn-orth}
    (\bI_n- \bP_{-j,d})\bP_d\bU_n = (1 - \lambda_{d})(\bI_n- \bP_{-j,d})\bP_d\bU_{n,>d}
\end{equation}

Using this, for $j \in O_{X,d}$, we can simplify the bias of the estimator as:
\begin{equation}\label{eq:bias-drop}
    \begin{split}
        & \hat{\beta}_{-j,n} - \beta\\
        = & \dfrac{\frac{1}{n}\bX_{n}^\top\bP_d(\bI_n - \bP_{-j,d})\bP_d(\bU_n + \beps_n)}{\frac{1}{n}\bX_{n}^\top\bP_d(\bI_n - \bP_{-j,d})\bP_d\bX_{n}}\\
        = & \dfrac{\frac{1}{n}\bX_{n}^\top\bP_d(\bI_n - \bP_{-j,d})\bP_d(\bU_{n,d} + (1 - \lambda_{d})\bU_{n,>d} + \beps_n)}{\frac{1}{n}\bX_{n}^\top\bP_d(\bI_n - \bP_{-j,d})\bP_d\bX_{n}}\\
        \overset{\ref{eq:projn-orth}}{=} & \dfrac{\frac{1}{n}\bX_{n}^\top\bP_d(\bI_n - \bP_{-j,d})\bP_d((1 - \lambda_{d})\bU_{n,>d} + \beps_n)}{\frac{1}{n}\bX_{n}^\top\bP_d(\bI_n - \bP_{-j,d})\bP_d\bX_{n}}\\
        = & \dfrac{\frac{1}{n}\bX_{n}^\top\bH_{d} \bA_{j,d}\frac{1}{n}\bH_{d}^\top((1 - \lambda_{d})\bU_{n,>d} + \beps_n)}{\frac{1}{n}\bX_{n}^\top\bH_{d} \bA_{j,d}\frac{1}{n}\bH_{d}^\top\bX_{n}}\\
        = & \underbrace{\dfrac{\frac{1}{n}\bX_{n}^\top\bH_{d} [A^{lim}_{j,d} + (\bA_{j,d} - A^{lim}_{j,d})]\frac{1}{n}\bH_{d}^\top((1 - \lambda_{d})\bU_{n,>d} + \beps_n)}{\balp_{x,d}^\top A^{lim}_{j,d}\balp_{x,d}}}_{\coloneqq L_n}\cdot \underbrace{\frac{\balp_{x,d}^\top A^{lim}_{j,d}\balp_{x,d}}{\frac{1}{n}\bX_{n}^\top\bH_{d} \bA_{j,d}\frac{1}{n}\bH_{d}^\top\bX_{n}}}_{\coloneqq \zeta_{j,n}\pnlim 1}\\
    \end{split}
\end{equation}
Thus, note that the asymptotic distribution of  $\sqrt{n}\left(\hat{\beta}_{-j,n} - \beta\right)$ is same as asymptotic distribution of $\sqrt{n}L_n$. Now:
$$
\begin{aligned}
    & \sqrt{n}L_n \cdot \balp_{x,d}^\top A^{lim}_{j,d}\balp_{x,d}\\
    =& \sqrt{n}\frac{1}{n}\bX_{n}^\top\bH_{d} [A^{lim}_{j,d} + (\bA_{j,d} - A^{lim}_{j,d})]\frac{1}{n}\bH_{d}^\top((1 - \lambda_{d})\bU_{n,>d} + \beps_n)\\
    = & \underbrace{\sqrt{n}\left[\frac{1}{n}\bX_{n}^\top\bH_{d}A^{lim}_{j,d}\frac{1}{n}\bH_{d}^\top((1 - \lambda_{d})\bU_{n,>d} + \beps_n)\right]}_{R_{n}}\\
    &+\underbrace{\left(\frac{1}{n}\bX_{n}^\top\bH_{d}\right)\sqrt{n}(\bA_{j,d} - A^{lim}_{j,d})\left(\frac{1}{n}\bH_{d}^\top((1 - \lambda_{d})\bU_{n,>d} + \beps_n)\right)}_{T_{n}}
\end{aligned}
$$
Next, we show that $T_{n} = \smallo_P(1)$. In order to see that, notice that from Equation \ref{eq:X-Hk-inner} we have, $\frac{1}{n}\bX_{n}^\top\bH_{d} \pnlim \balp_{x,d}$ and thus it is $\bigo_P(1)$. Also using Lemma \ref{lem:dot-prod-lim}, we can say that $\frac{1}{n}\bH_{d}^\top((1 - \lambda_{d})\bU_{n,>d} + \beps_n) \pnlim \bzero_d$, Hence we have using Lemma \ref{lem:A-lim} that:
$$
\begin{aligned}
    |T_{n}| = & \left|\left(\frac{1}{n}\bX_{n}^\top\bH_{d}\right)\sqrt{n}(\bA_{j,d} - A^{lim}_{j,d})\left(\frac{1}{n}\bH_{d}^\top((1 - \lambda_{d})\bU_{n,>d} + \beps_n)\right)\right|\\
    \leq & \underbrace{\left|\left|\frac{1}{n}\bX_{n}^\top\bH_{d}\right|\right|_2}_{\bigo_P(1)}\underbrace{\left|\left|\sqrt{n}(\bA_{j,d} - A^{lim}_{j,d})\right|\right|_2}_{\bigo_P(1)}\underbrace{\left|\left|\frac{1}{n}\bH_{d}^\top((1 - \lambda_{d})\bU_{n,>d} + \beps_n)\right|\right|_2}_{\smallo_P(1)}\\
    = & \smallo_P(1)
\end{aligned}
$$
Thus, we can say that $\sqrt{n}\left(\hat{\beta}_{-j,n} - \beta\right)$ has the same asymptotic distribution as that of the quantity $\sqrt{n}L_n = R_{n}\left(\balp_{x,d}^\top A^{lim}_{j,d}\balp_{x,d}\right)^{-1}$. So let us start deriving the asymptotic distribution of $R_{n}$:
% For that, first let us denote the vector $\bh_{d,i} = (h_1(S_i), \cdots, h_d(S_i))^\top$. Hence we can write down $R_{n}$ as:
\begin{equation}\label{eq:Rn-clt-map}
    \begin{split}
            R_{n} = & \sqrt{n}\left[\frac{1}{n}\bX_{n}^\top\bH_{d}A^{lim}_{j,d}\frac{1}{n}\bH_{d}^\top((1 - \lambda_{d})\bU_{n,>d} + \beps_n)\right]\\
            = & \sqrt{n}\left[\frac{1}{n}\bX_{n}^\top\bh_{n,j}\frac{1}{n}\bh_{n,j}^\top((1 - \lambda_{d})\bU_{n,>d} + \beps_n)\right]\\
    = & \sqrt{n}\left[\underbrace{\left(\frac{1}{n}\sum_{i=1}^nX_ih_{j}(S_i)\right)}_{R_{1,n}}\underbrace{\left(\frac{1}{n}\sum_{i=1}^n((1 - \lambda_{d})U_{>d,i}+\epsilon_i)h_{j}(S_i)\right)}_{R_{2,n}}\right]\\
    = & \sqrt{n} R_{1,n} \cdot R_{2,n}
    \end{split}
\end{equation}
Let us define $\bZ_i = \begin{bmatrix}
     Z_{i1} \coloneqq X_ih_{j}(S_i)\\
     Z_{i2} \coloneqq((1 - \lambda_{d})U_{>d,i}+\epsilon_i)h_{j}(S_i)
 \end{bmatrix}$ which implies $\frac{1}{n}\sum_{i=1}^n\bZ_i = \begin{bmatrix}
     R_{1,n}\\
     R_{2,n}
 \end{bmatrix}$. 
 First, we will derive a CLT for $\bZ_i$'s as they are iid and then use the Delta Method on the sample averages of $\bZ_i$ to get the CLT for $R_n$. Let us first derive the mean vector $\bmu_Z$ for $\bZ_i$. Remember that $\epsilon_i \perp S_i$. Hence $\E X_ih_{j}(S_i) = \alpha^{(j)}_x$ and $\E((1 - \lambda_{d})U_{>d,i}+\epsilon_i)h_{j}(S_i) = 0$ which implies that
 $
 \bmu_Z^\top = \begin{bmatrix}
    \alpha^{(j)}_x & 0
 \end{bmatrix}
 $. 
 Now, let us look at the covariance matrix of $\bZ_i$ denoted by $\bSig_Z$. Thus, we have:
 $$
     \bSig_Z = \begin{bmatrix}
         \sigma^2_1 & \sigma_{12}\\
         \sigma_{12} & \sigma^2_2 \\
     \end{bmatrix}
 $$
For this Delta Method, we will only need $\sigma^2_2$, which we will compute. For that, notice that:
$$
\begin{aligned}    
\sigma^2_2= & Var((1 - \lambda_{d})U_{>d,i}+\epsilon_i)h_{j}(S_i))\\
= & \E[(1 - \lambda_{d})U_{>d,i}+\epsilon_i)h_{j}(S_i)]^2\\
= & (1 - \lambda_{d})\E[U_{>d,i}h_j(S_i)]^2 + \sigma^2_\epsilon\\
\overset{(\blue{*})}{=} & (1 - \lambda_{d}) Var[U_{>d,i}h_j(S_i)] + \sigma^2_\epsilon\\
\coloneqq & (1 - \lambda_{d}) \sigma^2_{U,j,>d} + \sigma^2_\epsilon
\end{aligned}
$$
Let us discuss in short why $(\blue{*})$ holds in the chain of inequality above which is essentially showing why $\E [U_{>d}(s)h_j(S)] = 0$. Note that $U_{>d}(s) = \sum_{j \geq d+1}\alpha_u^{(j)}h_j(s)$, and let us define $U_{>d,k}(s) \coloneqq \sum_{j=d+1}^k \alpha_u^{(j)}h_j(s)$. Then we have $h_j(s)U_{>d,k}(s) \klim h_j(s)U_{>d}(s)$ pointwise for every $s \in \calD$ and for $1\leq j\leq d$, we have:
$$
\begin{aligned}
\E h_j(S)U_{>d,k}(S) =\int h_j(s)X_{>d,k}(s)\,f_Sds = &\sum_{l=d+1}^k\alpha_u^{(l)}\int h_j(s)h_l(s)\,f_S\,ds\\
= & \sum_{l=d+1}^k\alpha_u^{(l)}\mathbb{I}(l=j)=0
\end{aligned}
$$
The reason for the last equality is that $l \geq d+1$ and $l \leq d$. Thus Dominated Convergence Theorem (DCT) implies that:
$$
\E h_j(S)U_{>d,k}(S) \klim \E h_j(S)U_{>d}(S) = \sum_{l=d+1}^\infty\alpha_u^{(l)}\mathbb{I}(l=j) = 0
$$
This shows why the expectation of $h_j(S)U_{>d}(S)$ is 0. Thus, using Multivariate CLT, we can say that:
\begin{equation}\label{eq:joint-clt}
    \sqrt{n}\left(\dfrac{1}{n}\sum_{i=1}^n \bZ_i - \bmu_Z\right) = \sqrt{n}\left(\begin{bmatrix}
     R_{1,n}\\
     R_{2,n}
 \end{bmatrix} - \begin{bmatrix}
     \alpha_x^{(j)}\\
     0
 \end{bmatrix}\right) \dnlim \mathcal{N}_{2}\left(\bzero_{2},\bSig_Z\right)
\end{equation}
Recall from Equation \ref{eq:Rn-clt-map}, we want to get to the CLT for $R_n$ for which we will apply the Delta Method on the CLT that we just derived. For a function $g: \R^{2} \mapsto \R$ defined by $g(x,y) = xy$. The values of the functions $g$  evaluated at $\bmu_Z$ is given by $g(\bmu_Z) = 0$.
The gradient of the function $g$ is then given by $\nabla g^\top = \begin{bmatrix}
    y & x
\end{bmatrix}$. The value of $\nabla g$ evaluated at $\bmu_Z$ is given by 
        $\nabla g^\top(\bmu_Z) = \begin{bmatrix}
    0 & \alpha_x^{(j)} 
\end{bmatrix}$.
This means that:
\begin{equation}\label{eq:lim-var}
    \begin{split}
        \sigma^2_R \coloneqq & (\nabla g(\bmu_Z))^\top\bSig_Z \nabla g(\bmu_Z) \\
        =& \alpha_x^{(j)^2} \sigma^2_2\\
        = & \alpha_x^{(j)^2}(\sigma^2_{U,j,>d} + \sigma^2_\epsilon)
    \end{split}
\end{equation}
Thus applying Delta Method on the Equation \ref{eq:joint-clt}, we have:
\begin{equation}\label{eq:Rn-clt}
\begin{split}
    \sqrt{n}\left[g\left(\dfrac{1}{n}\sum_{i=1}^n \bZ_i\right) - g\left(\bmu_Z\right)\right]& \dnlim N\left(0, (\nabla g(\bmu_Z))^\top \bSig_Z \nabla g(\bmu_Z)\right)\\
    \implies \sqrt{n}\left(g\begin{bmatrix}
     R_{1,n}\\
     R_{2,n}
 \end{bmatrix} - 0\right) & \dnlim N(0,\sigma^2_R)\\
 \implies \sqrt{n}\left( R_{1,n}\cdot R_{2,n}\right)& \dnlim N(0,\sigma^2_R)\\
    \implies R_n & \dnlim R \sim N(0, \sigma^2_R) \quad \text{(see Eqn. \ref{eq:lim-var})}
    \end{split}
\end{equation}
Hence, combining all the components from Equation \ref{eq:bias-drop}, we have for $j \in O_X$:
\begin{equation}
    \begin{split}
        \sqrt{n}\left(\hat{\beta}_{-j,n} - \beta\right) = & \sqrt{n}L_n \zeta_{j,n}\\
        = & \frac{R_n + T_n}{\balp_{x,d}^\top A^{lim}_{j,d}\balp_{x,d}}\zeta_{j,n}\\
        = & \frac{R_n + T_n}{\alpha_x^{(j)^2}}\zeta_{j,n} \dnlim \frac{R}{\alpha_x^{(j)^2}} \sim N\left(0,\frac{\sigma^2_R}{\alpha_x^{(j)^4}}\right) = N\left(0,\frac{(1 - \lambda_{d}) \sigma^2_{U,j,>d} + \sigma^2_\epsilon}{\alpha_x^{(j)^2}}\right)
    \end{split}
\end{equation}
This completes the proof on the asymptotic normality of the valid candidates.

Finally, observe that when $U$ has a finite expansion with $d$ basis functions, i.e. when $\lambda_d = 1$, then for $j \in O_{X,d}$ recalling that  error vector $\beps_n$ is independent of the sampled location vector $\bS_n$, from Equation \ref{eq:bias-drop}, we have:
\begin{equation}
    \begin{split}
        \E\left[\hat{\beta}_{-j,n} - \beta\right] = & \E_{\beps_n,\bS_n}\left[\dfrac{\frac{1}{n}\bX_{n}^\top\bH_{d} \bA_{j,d}\frac{1}{n}\bH_{d}^\top\beps_n}{\frac{1}{n}\bX_{n}^\top\bH_{d} \bA_{j,d}\frac{1}{n}\bH_{d}^\top\bX_{n}}\right]\\
        = & \E_{\bS_n}\left[\E\left[\dfrac{\frac{1}{n}\bX_{n}^\top\bH_{d} \bA_{j,d}\frac{1}{n}\bH_{d}^\top\beps_n}{\frac{1}{n}\bX_{n}^\top\bH_{d} \bA_{j,d}\frac{1}{n}\bH_{d}^\top\bX_{n}}\Big| \bS_n\right]\right]\\
        = & \E_{\bS_n}\left[\dfrac{\frac{1}{n}\bX_{n}^\top\bH_{d} \bA_{j,d}\frac{1}{n}\bH_{d}^\top}{\frac{1}{n}\bX_{n}^\top\bH_{d} \bA_{j,d}\frac{1}{n}\bH_{d}^\top\bX_{n}}\E\left[\beps_n|\bS_n\right]\right]\\
        = & \E_{\bS_n}\left[\dfrac{\frac{1}{n}\bX_{n}^\top\bH_{d} \bA_{j,d}\frac{1}{n}\bH_{d}^\top}{\frac{1}{n}\bX_{n}^\top\bH_{d} \bA_{j,d}\frac{1}{n}\bH_{d}^\top\bX_{n}}\E\left[\beps_n\right]\right]\\
        = & 0
    \end{split}
\end{equation}
This proves the finite sample unbiasedness property of our valid drop-one candidate estimator under finite expansion assumption for the unmeasured confounder $U$. This completes the proof of Theorem \ref{thm:drop-est-prop}.
\end{proof}
\subsection{Proof of Theorem \ref{thm:asm-var-comparison}}\label{thm:asm-var-comparison-proof}
\begin{proof}
Recall from Theorems \ref{thm:IV-est-prop} and \ref{thm:drop-est-prop}, the asymptotic variance of the drop-one 
estimator for $j \in O_{X,d}$ is $avar\left(\hat{\beta}_{-j,n}\right)=\dfrac{(1-\lambda_{d})\sigma^2_{U,j,>d} + \sigma^2_\epsilon}{{\alpha_x^{(j)}}^2}$ and that for the projection estimator is $avar\left(\hat{\beta}^{PROJ}_{j,n}\right)=\frac{\sigma^2_{U,j}+ \sigma^2_\epsilon}{{\alpha_x^{(j)}}^2}$. Let us compare the asymptotic variances under two scenarios:

\textbf{\textit{Case 1: $(\lambda_d = 1)$}} This scenario implies that with respect to the basis function used for constructing the estimator, the unmeasured confounder assumes a finite expansion i.e. $ U(s) = \sum_{j=1}^d\alpha_u^{(j)}h_j(s)$ which gives us $avar\left(\hat{\beta}_{-j,n}\right) = \sigma^2_\epsilon/{{\alpha_x^{(j)}}^2}$. Note that, excluding pathological cases, $\sigma^2_{U,j}=\int_\calD (h_j(s)U(s))^2f_S(s)ds  = \text{Var}(h_j(S)U(S)) > 0$ since $\alpha_u^{(j)}=0$. So the asymptotic variance of $\hat{\beta}_{j,n}^{PROJ}$ is strictly greater than $\sigma^2_\epsilon/{\alpha_x^{(j)}}^2$. This completes the proof of this case.

\textbf{\textit{Case 2: $(\lambda_d = 0)$}}: For this case we will use Lemma \ref{lem:truncated-variance}. Since the basis functions have been assumed to continuous and the domain $\calD$ is compact, Assumption (\ref{asm:infty-bound-basis}) of the lemma is immediately satisfied. Assumption (\ref{asm:c-star}) is also satisfied for this basis function class for a choice of very small $\epsilon$ which corresponds a large enough $d$ since the unmeasured confounder is assumed to be a continuous function in the compact domain $\calD$. Thus for $j \in O_{X,d}$, invoking the lemma we have $\|h_j U\|^2_{f_S} > \|h_j U_{>d}\|^2_{f_S}$. Few things to observe here is that $\|h_j U\|^2_{f_S} = \int_\calD (h_jU)^2f_Sds = \sigma^2_{U,j}$. Next we see that $\sigma^2_{U,j,>d} = \int_\calD(h_j U_{>d})^2f_Sds= \|h_j U_{>d}\|^2_{f_S}$. Thus, using this lemma, it can be directly inferred from Theorems \ref{thm:IV-est-prop} and \ref{thm:drop-est-prop} that $\sigma^2_{U,j} > \sigma^2_{U,j,>d}$ for $j \in O_{X,d}$. Thus even for $U$ having possibly infinite expansion i.e. for $\lambda_{d} = 0$, if $\calH$ satisfies the assumptions of Theorem \ref{thm:asm-var-comparison} then:
$$ avar\left(\hat{\beta}_{-j,n}\right)
< avar\left(\hat{\beta}^{PROJ}_{j,n}\right).$$
This completes the proof of the theorem.   
\end{proof}
\end{document}